\documentclass[12pt]{article}

\usepackage{amsmath,amsfonts,amssymb,amsthm}
\usepackage{setspace,graphicx}
\usepackage{natbib}
\usepackage{bm}
\usepackage{ifthen}
\usepackage{verbatim}
\usepackage{color}
\usepackage{hyperref}
\usepackage{dsfont}
\usepackage{tikz}
\usepackage{setspace,graphicx}
\usepackage{sgame}
\usepackage{bm}
\usepackage{multirow}
\usepackage{ifthen}
\usepackage{verbatim}
\usepackage{xcolor}
\usepackage[multiple]{footmisc}

\usetikzlibrary{arrows.meta}

\theoremstyle{plain}
\theoremstyle{definition}
  \newtheorem{theorem}{Theorem}

  \newtheorem{corollary}{Corollary}
  \newtheorem{definition}{Definition}

  \newtheorem{lemma}{Lemma}
  
  \newtheorem{proposition}{Proposition}

  \theoremstyle{remark}

\def\EE{\mathbb{E}}
\def\RR{\mathbb{R}}
\DeclareMathOperator*{\argmax}{arg\,max}

\begin{document}


\title{Temporary exclusion in repeated contests}

\author{Yaron Azrieli\thanks{Department of Economics, The Ohio State University, 1945 North High street, Columbus, OH 43210, azrieli.2@osu.edu. I thank Nageeb Ali, Larry Ausubel, Simon Board, Yannai Gonczarowski, PJ Healy, Andrew McClellan, Jim Peck, Eran Shmaya, Ali Shourideh, Andy Skrzypacz, Ece Yegane, as well as participants at SITE dynamic games conference, Stony Brook game theory conference, Midwest theory conference, and UMD theory seminar for their comments and suggestions.}}

\maketitle


\bigskip
\bigskip

\begin{abstract}
Consider a population of agents who repeatedly compete for awards, as in the case of researchers annually applying for grants. Noise in the selection process may encourage entry of low quality proposals, forcing the principal to commit large resources to reviewing applications and further increasing award misallocation. A \emph{temporary exclusion} policy prohibits an agent from applying in the current period if they were rejected in the previous. We compare the steady state equilibria of the games with and without exclusion. Whenever the benefit from winning is sufficiently large exclusion results in more self-selection, eliminating entry of low quality applications. We extend the analysis to more general exclusion policies. We also show that exclusion has a distributional effect, where better able agents exhibit more self-selection. 




\end{abstract}

\newpage

\section{Introduction}\label{sec-introduction}

In many types of contests, a key challenge for the organizer is to efficiently allocate awards to the applications of highest quality. Consider for example the competition among scientists in a certain field for grants from a funding agency: The reviewing process only provides a noisy signal of the potential of each proposal, implying that some of the budget is misallocated.\footnote{For example, \citet{pier2018low} documents the high variance of evaluations when the same NIH grant proposal is reviewed by multiple reviewers.} Similar issues arise when an employer makes hiring decisions among a pool of job applicants, or when a company selects among proposed solutions in a crowdsourcing challenge. The noisiness of the allocation process provides incentives for agents with low quality applications to nevertheless `give it a shot', which in turn increases the volume of applications, forces the designer to commit more resources to the reviewing process, and potentially leading to even more allocation mistakes.\footnote{For a concrete example, consider the following quote from The Guardian newspaper (2011): ``At the height of last year's Gulf of Mexico oil disaster, BP invited members of the public to submit their ideas for sealing off its ruptured well and cleaning up the millions of barrels of oil that had leaked into water, marsh and beaches... Some 123,000 people responded from more than 100 countries... More than 100 experts reviewed more than 43,000 suggestions... Testing teams struggled to give people a fair hearing... The technical experts who reviewed the suggestions describe a giant sifting exercise that, in the end, yielded relatively little in the way of results.''} The incentive to enter the contest is particularly strong if the benefit that agents receive from winning is large.\footnote{Winning grants has a critical role in determining researchers' career prospects \citep{melkers202318}.} 

To mitigate these problems, the designer can try to encourage self-screening by increasing the cost of participation in the contest. For example, they can charge submission fees, require the applicant to meet certain criteria, or make the application process cumbersome. But such measures will likely fail if the prizes are high, and moreover may inadvertently `price out' the wrong applicants and discriminate against agents that have smaller budgets or lower ranks.\footnote{See for example \citet{oldford2023marginalizing} on the marginalizing effect that accounting and finance journals submission fees has on researchers with smaller research budgets, and \citet{cruz202317} which surveys the literature on gender and underrepresented minorities differences in research funding.}

In this paper we study an alternative way to encourage self-screening, taking advantage of the fact that an agent may wish to participate in similar future contests. Namely, we analyze policies according to which an agent may become temporarily ineligible to apply, depending on the history of their actions and outcomes. Such \emph{temporary exclusion} policies create a direct link between contests of different periods and introduce intertemporal incentives to an agent's participation decision.\footnote{We are aware of several institutions that use exclusion policies, presumably as a way to reduce the volume of applications and encourage self-selection. For example, the Graduate Research Fellowship Program (GRFP) of the NSF restricts students to apply at most once, see the \href{https://www.nsf.gov/pubs/2023/nsf23605/nsf23605.pdf}{2023 program solicitation}; leading management consulting firms such as Bain, BCG, and McKinsey ban rejected applicants from reapplying for a certain period of time, see the discussion in this \href{https://strategycase.com/ban-period-in-consulting-applications}{blog post}.}  Note that with exclusion the cost of participation is endogenous and depends both on the agent's own plan for future rounds as well as on the strategies employed by all other agents. A key insight of this paper is that, unlike with monetary submission fees, in the equilibrium with exclusion the cost of participation scales up linearly with the stakes of the contest, implying that low quality applicants self-select to stay out even in contests with arbitrarily high prizes. 


Our model is suitable for environments with the following key features. (i) A large population of agents repeatedly compete for awards; (ii) the quality of each agent's application is randomly drawn each period from a fixed distribution; (iii) there is a fixed budget of awards to be allocated each period and the designer's goal is to choose the highest quality applications; and (iv) the review process is such that higher quality applications are more likely to be selected but mistakes happen. Among these assumptions we view (ii) as the most restrictive since it rules out quality persisting across periods and the possibility of learning; it is however needed to keep the model tractable. 

We start by analyzing a benchmark model with no exclusion. The equilibrium is characterized by a cutoff quality $Q_0$ such that agents with qualities above $Q_0$ apply and those with qualities below $Q_0$ do not apply. This cutoff is the unique quality at which the expected payoff from applying, conditional on being the marginal quality to apply, is equal to the expected payoff from not applying (which is normalized to zero). In particular, as prizes increase more agents choose to apply and in the limit we get universal participation. 

We then move on to analyze a simple exclusion policy according to which an agent who gets rejected in a given period is excluded from participation in the next period, after which they are eligible again. We focus on this particular policy since it is straightforward to implement, relatively easy to analyze, and sufficient to illustrate the key trade-offs that arise with temporary exclusion. We look for `steady state equilibria' defined as stationary strategy profiles in which the quality distribution of applications remains constant over time, and in which each agent best responds to this recurrent distribution. Such equilibria are also characterized by a cutoff quality for entry, $Q_1$, but there are two key differences from the benchmark. On the one hand, the best-response of an agent to a given (stationary) quality distribution is more conservative than in the benchmark, i.e., exclusion indeed makes rejection more costly and encourages self-selection. But on the other hand, in a steady state a certain fraction of applications is rejected every period and hence that fraction of the population is ineligible. Thus, the competition that an agent faces is less intense, which could encourage lower quality applications. Each of these two forces may dominate, so both $Q_1>Q_0$ and $Q_0>Q_1$ are possible in general. However, we prove that whenever the benefit from winning is large enough the former is true, meaning that exclusion eliminates the left tail of the quality distribution. 

To evaluate the impact that exclusion has on the contest's outcome, we compare the quality distributions of applications and of winners between the benchmark equilibrium and the equilibrium under the above exclusion policy. Whenever $Q_1>Q_0$ (which is always the case for high prizes), the total volume of submissions with exclusion is lower, both because low quality applications between the two cutoffs are no longer submitted, and because some agents with high quality ideas above $Q_1$ are ineligible. The lower volume implies that agents' welfare is higher with exclusion. In terms of winners, exclusion eliminates funding for qualities below $Q_1$ since these are no longer being submitted, but at the same time reduces the share of the budget that is allocated to top-quality applications. Overall, in our model exclusion tends to shift budget from the two extremes to the middle range of qualities, see Proposition \ref{prop-noisy-review-comparison} for a formal result.



Exclusion policies can of course be generalized and the principal can tailor the policy to fit the particular characteristics of the contest environment. We study two such generalizations: In the first, applicants receiving low evaluations, below a bar set by the designer, are excluded in the next period. Thus, exclusion is not necessarily triggered by rejection. The analysis and results carry over to this kind of policies. Second, we analyze exclusion that lasts for more than one period. In Proposition \ref{prop-T-large} we show, somewhat unintuitively, that long bans lead to the principal's worst possible outcome in which almost all eligible agents apply (and win) every period.


Finally, our model assumes that all agents are identical in terms of the quality distribution they generate each period. In Section \ref{sec-heterogenous} we relax this assumption and show how to extend the analysis to heterogeneous populations. In particular, we prove that in a steady state equilibrium with exclusion agents that produce applications of (stochastically) higher quality are more conservative and use a higher entry cutoff. Thus, exclusion policies also induce interesting distributional effects in award allocation.   



\subsection{Literature review}\label{subsec-lit_rev} 

Our initial motivation comes from contests for research grants and for journal space (publications), see \citet{lepori2023handbook} for surveys of key issues in grant allocation.\footnote{As mentioned above, our model can also be applied to other types of contests, e.g., crowdsourcing contests. \citet{acar2019motivations} studies the challenges faced by organizations due to the high volume of low quality solutions submitted. See also \citet{chen2021attracting}, \citet{patel2023monetary}, and the references therein.} \citet{azar2006academic} discusses several ideas of how to make the academic review process more efficient, one of which is to encourage self-screening by limiting the number of rejections a given paper is allowed to receive; temporary exclusion that we study here is similar in spirit but applied at the researcher level rather than at the paper level. A recent paper by \citet{adda2024grantmaking} analyzes how noise in the review process and the formula for allocating funds between different research fields affect researchers' incentives to apply for grants; the special case of their model with a fixed budget is essentially the same as our benchmark. Other theory papers that study how features of the editorial process such as delay times and submission fees affect researchers' incentives to submit their work include \citet{leslie2005delays}, \citet{azar2007slowdown}, \citet{cotton2013submission}, and \citet{muller2021gatekeeper}. As far as we know, this is the first paper to consider repeated interaction in this context.

Temporary exclusion can be viewed as a way for the principal to link multiple decisions. There is substantial literature demonstrating the potential of such linkage to help principals achieve their goals, some prominent examples include \citet{casella2005storable} in the context of voting, \citet{jackson2007overcoming} in abstract social choice environments, and \citet{frankel2014aligned, frankel2016delegating, frankel2016discounted} for delegation problems. In all of these papers the principal does not get to observe any signal correlated with the agent's private information between periods (beyond the action taken by the agent), and the mechanism can be viewed as an ex-ante `budgeting' constraint on the number of times each action is taken. In contrast, the policies we consider condition eligibility on the outcome of the reviewing process.\footnote{One of the policies considered in subsection \ref{subsec-other-criteria} excludes all agents that applied in the previous period, regardless of the outcome of the evaluation process. The steady state quality distribution of winners under this policy is typically first-order dominated by the quality distribution of winners under policies that do condition eligibility on the signal obtained from the review.} Furthermore, our focus is on the steady states of the dynamic process induced by the mechanism, while the rest of this literature evaluates the performance of a mechanism by aggregating over the individual decision problems. See also Section \ref{sec-discussion} for further discussion.



There is a vast theoretical literature on contest design originating from the work of \citet{lazear1981rank}, where, unlike here, the goal of the designer is usually to maximize participants' effort. A main focus of this literature is the impact of the contest's prize structure on equilibrium efforts, e.g., \citet{moldovanu2001optimal, moldovanu2006contest}. \citet{olszewski2016large} show that the equilibria of contests with a large (but finite) number of participants can be approximated by the outcomes of a mechanism for a single agent with a continuum of types, and \citet{olszewski2020performance} apply this observation to derive the effort-maximizing prize structure in large contests. Several papers \citep{baye1993rigging, fullerton1999auctionin} show that higher effort can be induced by admitting only a subset of agents into the competition (i.e., exclude the others).\footnote{\citet{szymanski2003economic} surveys the literature on admission into sports competitions. See also the chapter on exclusion in \citet{konrad2009strategy}.} 

Much less work exists on repeated contests. \citet{meyer1991learning} studies the design of multi-stage contests in which the goal of the designer is to identify the better of two agents. \citet{krahmer2007equilibrium} analyzes infinitely repeated contests in which players learn about their own (fixed) ability over time. Unlike in these and similar papers, in our model agents' types are independently redrawn each period. \citet{de2021selecting} consider a principal that every period selects one of two agents to perform a task, and where the agents' private types are iid across time. They give a necessary and sufficient condition for the principal to be able to achieve the first-best payoff in equilibrium, and describe a principal's strategy that obtains the first-best. In our setup with many agents and no aggregate uncertainty the first-best can't be achieved since agents with lower types cannot be incentivized to stay out of the contest given that only the highest types participate.






\section{Preliminaries}\label{sec-preliminaries}

Throughout the paper we use the terminology of researchers competing for grants even though, as explained above, the model may describe other types of contests as well. 

\subsection{Environment}\label{subsec-environment}

Consider a large population of researchers, each producing one research idea per period. The quality (say, the potential benefit to society) of an idea is denoted by $q\in \RR$ and is assumed to be private information of the researcher.\footnote{All of our results go through if $q$ is a noisy estimate that a researcher has about the true underlying quality of their idea. However, such modification would change the interpretation of some of the results.} Each researcher chooses whether to submit their idea to be considered for funding by a central agency.
If an idea is submitted and funded then the benefit to the researcher is $V>0$; if the idea is submitted but not funded then the researcher suffers a loss of $C>0$; if a researcher chooses not to apply then their payoff is zero. Researchers discount future payoffs at a rate of $\delta\in (0,1)$. The agency has a total budget sufficient to fund a volume of $k\in (0,1)$ proposals.

Let $F(q)$ be a cumulative distribution function (cdf) of quality that admits a continuous and strictly positive density $f(q)$. We interpret $F$ as both the per-period distribution of quality in the population and as the belief of each researcher about the quality of ideas they produce each period. In particular, the overall quality distribution remains constant over time, all researchers are ex-ante identical, and there is no correlation across time in the quality that a researcher produces.

Let 
$$\Gamma_f = \{\phi:\RR\to \RR_+ ~:~ \phi \textit{ is measurable and } \phi\le f \}$$
be the set of non-negative functions bounded above by $f$. Each $\phi\in \Gamma_f$ can be thought of as a potential `density' of qualities of submitted ideas corresponding to some subset of the population choosing to apply for funding. We refer to elements of $\Gamma_f$ as quality distributions even though their integral is typically less than one; the volume of submissions under $\phi$ is $v(\phi):=\int \phi(q)dq$. We measure the distance between $\phi,\phi' \in \Gamma_f$ by $d(\phi,\phi')=\int|\phi(q)-\phi'(q)|dq$.\footnote{As is standard, we do not distinguish between $\phi$ and $\phi'$ if they are equal almost everywhere, so $d$ is a metric on $\Gamma_f$. This should not cause any confusion. Also, if a sequence $\{\phi_n\}\subseteq \Gamma_f$ converges almost everywhere to $\phi\in\Gamma_f$ then by Lebesgue's dominated convergence theorem we have that $d(\phi_n,\phi)\to 0$.} 

The probability that an application is successful depends on its own quality and on the overall quality distribution of submissions (the competition). We describe this probability by the function $W:\RR\times \Gamma_f \to [0,1]$, where $W(q,\phi)$ is the likelihood that an idea of quality $q$ gets funded given competition $\phi$.\footnote{The competition $\phi$ may affect the probability of success directly by crowding out proposals and indirectly by influencing the evaluation process.} We view $W$ as exogenously given and make the following assumptions:

\textbf{A1} $W$ is jointly continuous in $(q,\phi)$.

\textbf{A2} If $v(\phi)\le k$ then $W(q,\phi)=1$ for all $q$, and if $v(\phi)>k$ then $\int \phi(q)W(q,\phi)dq = k$.

\textbf{A3} If $v(\phi)>k$ then $\lim_{q\downarrow -\infty} W(q,\phi)=0$, $\lim_{q\uparrow +\infty} W(q,\phi)=1$, and $W(q,\phi)$ is strictly increasing in $q$.

\textbf{A4} If $\phi_1\le \phi_2$ then $W(q,\phi_1)\ge W(q,\phi_2)$ for each $q$.

\medskip

Assumption \textbf{A1} is for technical reasons. \textbf{A2} express the idea that the agency either funds all the applications (when the volume of application is less than the budget), or uses all its budget (when the volume of application is higher than the budget); indeed, $\phi(\cdot)W(\cdot,\phi)$ can be interpreted as the density of submissions being funded given $\phi$, so the integral is the overall volume funded. \textbf{A3} says that very bad ideas are unlikely to get funded and that the opposite is true for very good ideas; it also requires that higher quality is associated with higher chances of funding. Finally, \textbf{A4} implies that an application of a given quality has higher chances of being funded if there is less competition (fewer submissions of every quality). See subsection \ref{subsec-noisy-review} below for an example of a class of $W$ functions that satisfy these assumptions and can be interpreted as the result of a noisy review process of proposals.  



For any $Q\in \RR$ we denote by $f^Q\in \Gamma_f$ the quality distribution of submissions obtained when only ideas with quality at least $Q$ are submitted, namely,  
\begin{eqnarray*}
f^Q(q) =
\left\{ \begin{array}{ll}
0 & \textit{ if } ~~q<Q,\\
f(q) & \textit{ if } ~~q\ge Q.\\
\end{array} \right.
\end{eqnarray*}
We also allow for $Q=+\infty$ in which case $f^Q\equiv 0$ and $Q=-\infty$ in which case $f^Q=f$.

\subsection{Evaluating policies and the first-best}\label{subsec-evaluating}

Below we will study the effects that various exclusion policies have on the equilibrium outcomes of the contest. We choose not to commit to a particular objective function for the designer since their goals may widely vary depending on the context. Instead, we use various measures to compare policies.  

The equilibrium under a given policy can be summarized by the resulting quality distribution of submissions $\phi$. When combined with the (exogenously given) noise function $W$, one obtains the quality distribution of winners $\phi(\cdot)W(\cdot,\phi)$. The designer may wish to maximize the mean quality among winners, or, more generally, the mean of some increasing function of quality. Note that if the volumes of submissions $v(\phi), v(\phi')$ are both larger than $k$ then by assumption \textbf{A2} we have $\int \phi(q)W(q,\phi)dq = k =\int \phi'(q)W(q,\phi')dq$. Thus, we can use standard notions of stochastic dominance to compare such quality distributions. In particular, $\phi(\cdot)W(\cdot,\phi)$ first-order stochastically dominates $\phi'(\cdot)W(\cdot,\phi')$ if and only if the mean of any increasing function of quality is higher in the former.

The volume of submissions $v(\phi)$ may also enter the designer's objective. As discussed above, high volume likely requires costly resources to be used in the selection process. Thus, $\phi$ may be preferred over $\phi'$ if $v(\phi)$ is significantly smaller than $v(\phi')$, even if $\phi'$ induces a somewhat better quality distribution of winners.

Finally, agents' welfare may also be important for the principal. Since the volume of winners is always $k$ (so long as $v(\phi)\ge k$), welfare in each period is simply $kV-(v(\phi)-k)C$. Therefore, lower volume of submissions also implies higher welfare.

Let $Q^*=F^{-1}(1-k)$ be the top $k$ percentile of the quality distribution. We refer to $f^{Q^*}$ as the socially optimal or first-best quality distribution of submissions and to $Q^*$ as the socially optimal or first-best cutoff. Indeed, when submissions are given by $\phi=f^{Q^*}$ the quality distribution of funded ideas is $W(\cdot,f^{Q^*})f^{Q^*} = f^{Q^*}$ (by \textbf{A2}), which first-order stochastically dominates any other feasible distribution. Further, with $f^{Q^*}$ no submissions are rejected implying that reviewing resources are minimal and that researchers' welfare is maximized. 

\subsection{Probability of success as the outcome of noisy review}\label{subsec-noisy-review}

For some of the results and examples below we impose more structure on the function $W$. Let $G$ be a cdf that admits a continuous and everywhere positive density $g$. Suppose that the evaluation of an idea of quality $q$ generates a random signal $s\in \RR$ that is distributed according $G(s|q)=G(s-q)$. An idea gets accepted if the signal it generates exceeds a cutoff $\bar s$ that depends on the quality distribution of submissions $\phi$. 
 
Specifically, if $\phi\in \Gamma_f$ satisfies $v(\phi)>k$ then there is a unique cutoff signal $\bar s(\phi)$ such that\footnote{Indeed, when $\bar s \downarrow -\infty$ the integral converges to $v(\phi)>k$, and when $\bar s \uparrow +\infty$ the integral converges to $0$. It is also clear that the integral strictly decreases in $\bar s$ (follows from $G$ being strictly increasing) and that it is continuous in $\bar s$.}
\begin{equation}\label{eqn-signal-cutoff}
    \int \phi(q) (1-G(\bar s(\phi)|q)) dq =k.
\end{equation}
For every $\phi\in \Gamma_f$ with $v(\phi)>k$ and for every $q$ define $W(q,\phi) = 1-G(\bar s(\phi)|q)$, where $\bar s(\phi)$ is defined as the solution to (\ref{eqn-signal-cutoff}) (and define $W(q,\phi) =1$ for all $q$ whenever $v(\phi)\le k$). It is not hard to check that $W$ satisfies assumption \textbf{A1-A4} above.\footnote{One may also add the assumption that $\log(g)$ is concave to guarantee that signals satisfy the Monotone Likelihood Ratio Property, and hence that setting a lower bar for acceptance is justified. We do not require this as it is not needed for the analysis below.} 

When $W$ is obtained from some $G$ as described above we say that it is a \emph{Random Evaluation Function (REF)}.\footnote{\citet{azrieli2024success} axiomatize this class of contest success functions; \citet{morgan2018ponds} and \citet{adda2024grantmaking} use them in their models.}  The special case in which both $F$ and $G$ are normal distributions is referred to as the `normal-normal model'.

\section{Benchmark without exclusion}\label{sec-benchmark}

We start by analyzing a benchmark where researchers are free to submit their ideas every period regardless of their submission and outcome history. Since in this case there is no direct link between periods, we focus on `stationary equilibria' in which researchers choose the same submission strategy in every period.\footnote{This is justified in our context of a large population of researchers, since more sophisticated history-dependent equilibrium profiles may be hard to sustain. See \citet{green1980noncooperative} and \citet{sabourian1990anonymous} for formal arguments in this spirit.}

Suppose that the per-period quality distribution of submissions is $\phi\in \Gamma_f$. A researcher that has an idea of quality $q$ and chooses to submit receives an expected payoff of $W(q,\phi)V-(1-W(q,\phi))C$, while choosing not to submit results in a sure payoff of zero. We distinguish between two possible cases: First, if $v(\phi)\le k$ then by \textbf{A2} $W(q,\phi)=1$ for all $q$ so that the unique best response to $\phi$ is to always submit. Second, if $v(\phi)>k$ then by \textbf{A3} $W(q,\phi)$ is strictly increasing and converges to zero and one at negative infinity and at infinity, respectively. Further, by \textbf{A1} $W$ is continuous in $q$. Thus, there is a unique $Q\in \RR$ satisfying $W(Q,\phi)V-(1-W(Q,\phi))C= 0$, so that the best response is to submit if $q>Q$ and not to submit if $q<Q$ (at the threshold $Q$ researchers are indifferent; for expositional reasons we assume throughout that the tie-breaking rule favors submission). 

Since in equilibrium all researchers best respond to the same $\phi$, and since we just argued that the best response to any $\phi$ is unique, we may without loss restrict attention to symmetric equilibria where all researchers use the same submission strategy. If we identify the strategy `never submit' with the cutoff $Q=+\infty$, and the strategy `always submit' with the cutoff $Q=-\infty$, then we may also restrict attention to equilibria in cutoff strategies.

Now, if all researchers use cutoff $Q$, then the resulting quality distribution of submissions is $\phi=f^Q$. This leads to the following definition.

\begin{definition}\label{def-benchmark-eq}
A cutoff $Q_0\in \RR \cup \{\pm \infty\}$ is an equilibrium of the benchmark case if the strategy defined by $Q_0$ is the best response to $f^{Q_0}$ 
\end{definition}


\begin{proposition}\label{prop-benchmark}
There exists a unique equilibrium cutoff $Q_0$. This equilibrium satisfies $-\infty< Q_0< Q^*$ and is characterized by the equation 
\begin{equation}\label{eqn-benchmark}
    W(Q_0,f^{Q_0})=\frac{C}{C+V}. 
\end{equation} 
\end{proposition}


All proofs not in the main text appear in Appendix \ref{sec-proofs}. Since the cutoff $Q_0$ is smaller than the first-best $Q^*$, in equilibrium there is excess supply of submissions, $v(f^{Q_0})=1-F(Q_0)>1-F(Q^*)=k$, and the resulting quality distribution of funded ideas $f^{Q_0}(\cdot)W(\cdot,f^{Q_0})$ is first-order dominated by $f^{Q^*}(\cdot)W(\cdot,f^{Q^*})=f^{Q^*}$. From the researchers' point of view, the equilibrium per-period welfare is $kV-(1-F(Q_0)-k)C$, strictly smaller than $kV$ achieved at the first-best. See Figure \ref{figure-Benchmark} for an illustration in the normal-normal model.

As $V$ increases (or $C$ decreases) the equilibrium cutoff $Q_0$ decreases and the inefficiency becomes more severe. In the limit when $V$ grows to infinity, there is no self-selection and we converge to universal participation: The quality distribution of submissions converges to $f$, and the quality distribution of funded ideas converges to $W(\cdot,f)f(\cdot)$.

\begin{figure}
\centering
\includegraphics[width=.33\textwidth]{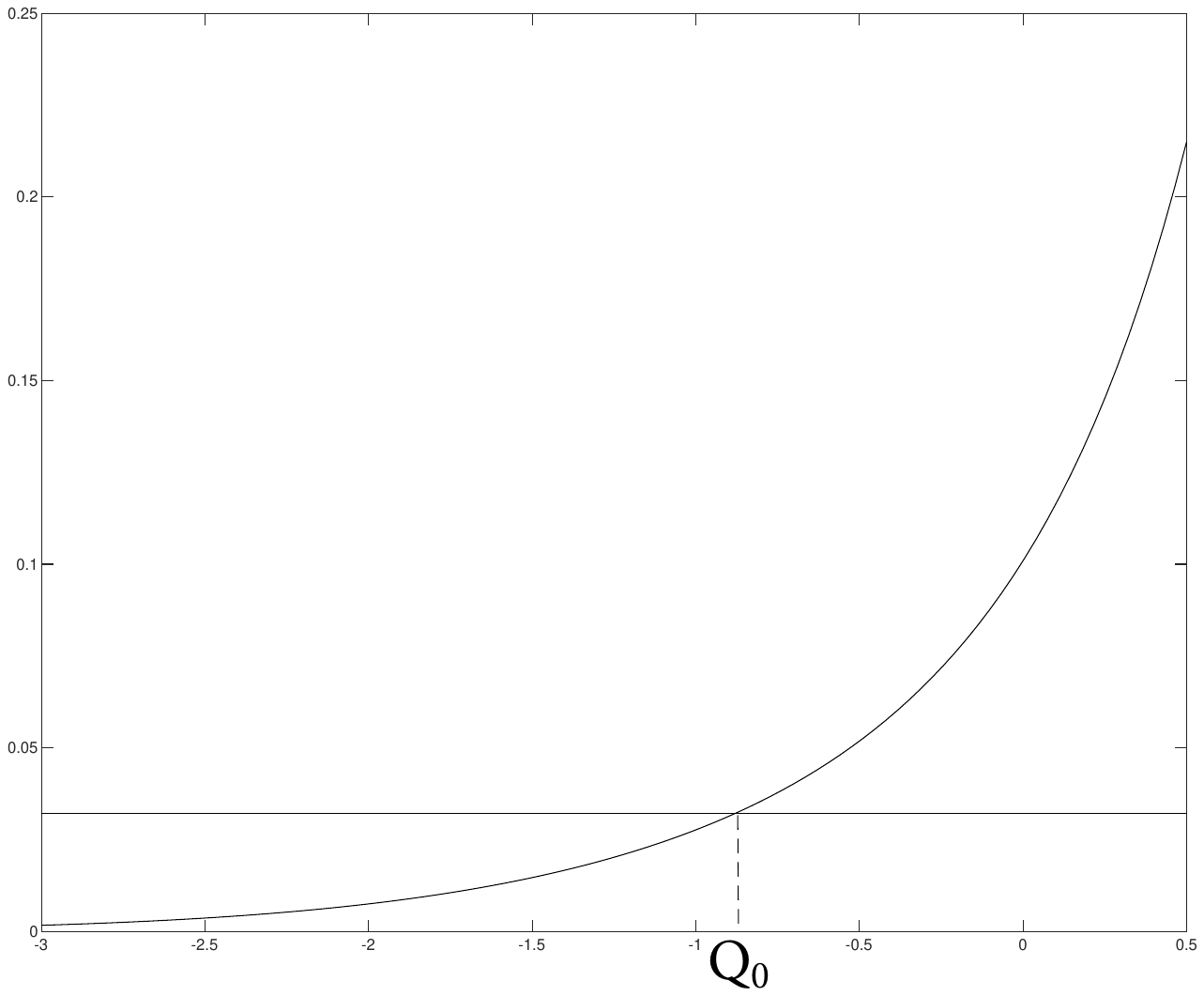}\hfill
\includegraphics[width=.33\textwidth]{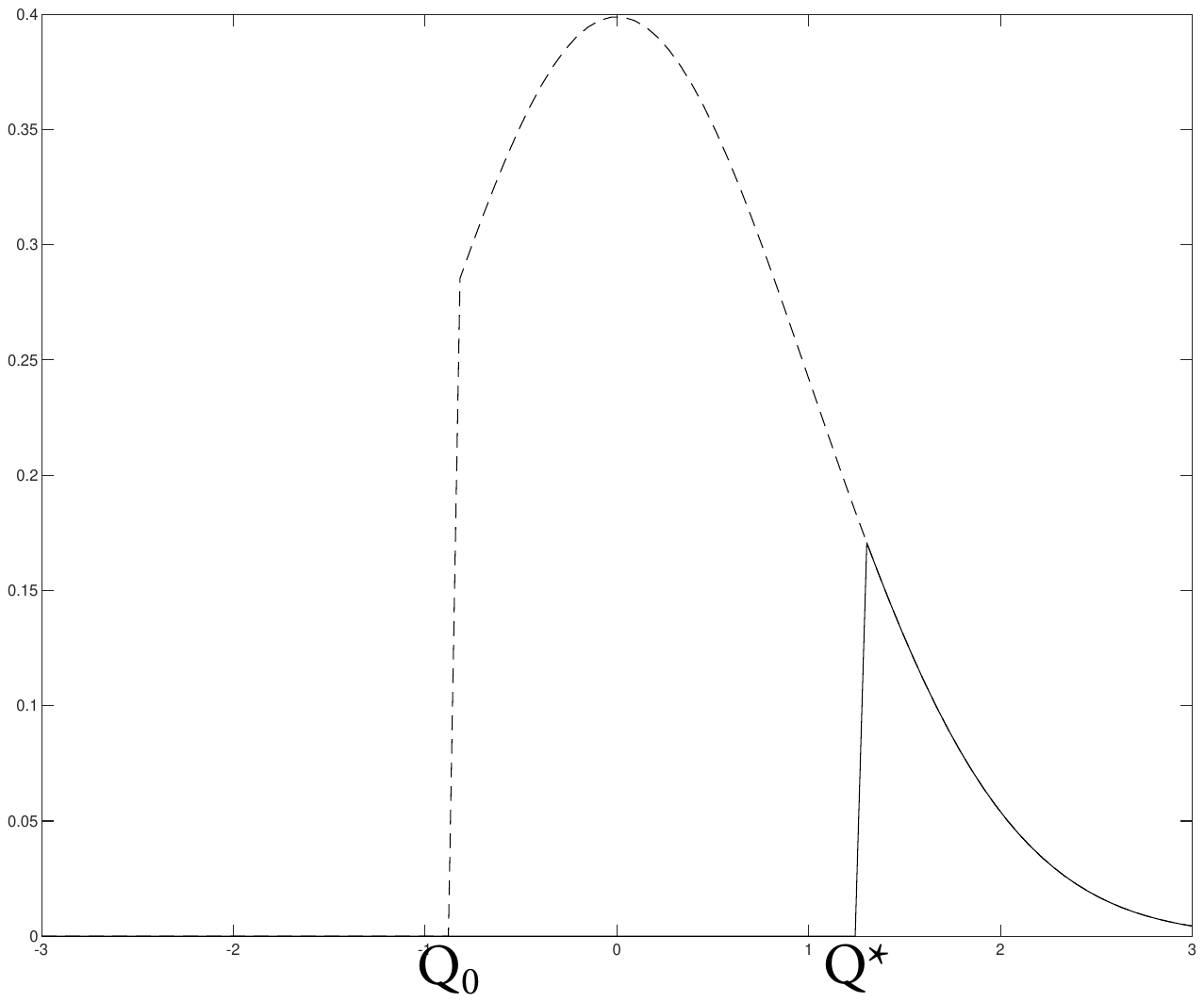}\hfill
\includegraphics[width=.33\textwidth]{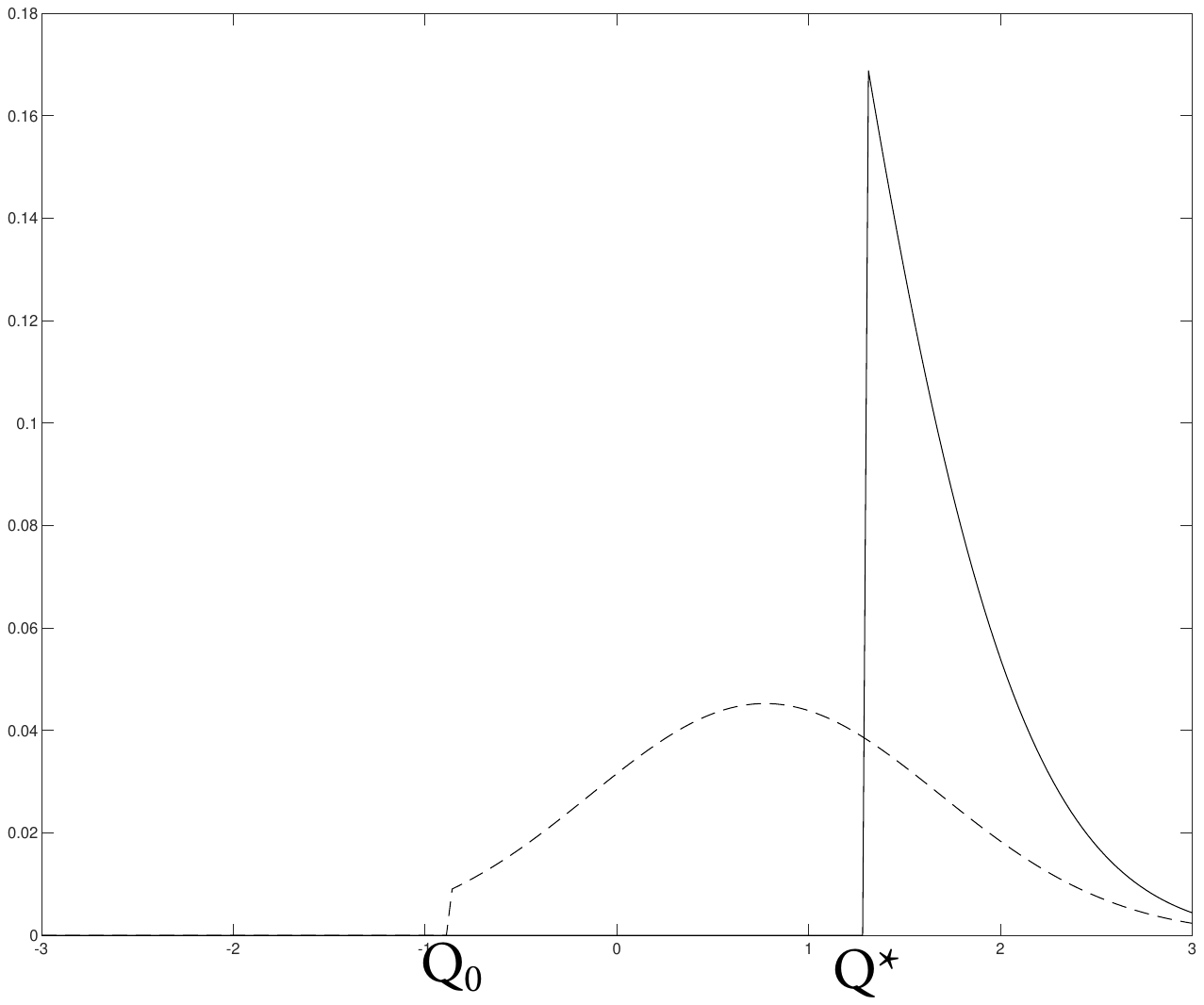}
\caption{The figure compares the first-best with the benchmark equilibrium in the normal-normal model with parameters $\mu_q=0$, $\sigma^2_q=1$, $\sigma^2_s=2$, $C=1$, $V=30$, and $k=0.1$. The left panel graphs the two sides of equation (\ref{eqn-benchmark}), so that the intersection point is the benchmark equilibrium cutoff $Q_0$; the middle panel shows the induced quality distribution of submissions in the benchmark equilibrium (dashed) and in the first-best (solid); the right panel shows the corresponding quality distributions of funded ideas.}
\label{figure-Benchmark}
\end{figure}



\section{Equilibrium with exclusion}\label{sec-steady-states-conditional}

Consider now a policy according to which a researcher who submitted their idea and got rejected is excluded from participation in the following period, after which they are eligible again. We focus on this specific exclusion policy because it is simple to implement and captures the key trade-offs involved; in the next section we consider more general policies. With such policy in place there is a direct link between periods, and submission decision of a researcher depends also on their expected payoff in future rounds. We study `steady state equilibria' of this dynamic game. In a steady state equilibrium the distribution of quality of submissions $\phi$ remains constant over time, and each researcher maximizes their life-time expected payoff.

\subsection{Best response}
We start by formalizing the utility maximization problem of a researcher who anticipates that a given quality distribution $\phi$ is submitted each period. We focus on stationary cutoff strategies characterized by a threshold $Q$, such that in every period the researcher submits their idea if $q\ge Q$ and does not submit if $q<Q$.\footnote{The restriction to these kind of strategies is without loss of generality in the following sense: First, it follows from \citet[Theorem 7]{blackwell1965discounted} that there is an optimal stationary strategy, so that the researcher cannot increase their utility by using history-dependent or time-dependent strategies. Second, it is not hard to show using assumption \textbf{A3} that an optimal stationary strategy must be a cutoff strategy. We omit the details.} As before, $Q=\pm\infty$ are allowed. Let $x(Q,\phi)$ be the total (current and future) expected payoff of a researcher who is eligible to submit in the current period and follows the stationary strategy with cutoff $Q$, assuming that $\phi$ repeats forever. Formally, $x(Q,\phi)$ is the solution of the equation
\begin{equation}\label{eqn-steady-state-payoff}
    x = F(Q)(0+\delta x) + \int_Q^{\infty} f(q)\left[W(q,\phi)(V+\delta x) + (1-W(q,\phi))(-C+\delta^2 x) \right] dq.
\end{equation}    

To interpret this equation, when $q<Q$ (which happens with probability $F(Q)$) the researcher does not submit and therefore gets a payoff of zero in the current period and will be eligible in the next period; when $q>Q$ the researcher does submit, wins with probability $W(q,\phi)$ for a payoff of $V$ today and continued eligibility in the next period, or gets rejected with probability $1-W(q,\phi)$ for a loss of $C$ today and will also have to sit out the next period.  

To simplify the notation we let $Win(Q,\phi) = \int_Q^{\infty} f(q)W(q,\phi)dq$ be the ex-ante winning probability of the researcher. 
Solving (\ref{eqn-steady-state-payoff}) for $x$ we get 
\begin{equation}\label{eqn-steady-state-payoff2}
x(Q,\phi) =\frac{Win(Q,\phi)V-(1-F(Q)-Win(Q,\phi))C}{(1-\delta)[1 + \delta(1-F(Q)-Win(Q,\phi))]}.
\end{equation}


Say that $Q$ is a best response to $\phi$ if $Q\in \argmax_{Q'} x(Q',\phi)$. The next lemma uses a version of the `one-shot deviation principle' to characterize the best response function.

\begin{lemma}\label{lemma-best-response}
(i) If $v(\phi)\le k$ then $Q=-\infty$ (always submit) is the unique best response.\\
(ii) If $v(\phi)> k$ then there is a unique best response threshold $Q\in \RR$. Moreover, $Q$ is the best response to $\phi$ if and only if 
    \begin{equation}\label{eqn-best-response}
       x(Q,\phi)> 0 ~~~~\textit{and}~~~~ W(Q,\phi)= \frac{C+\delta(1-\delta)x(Q,\phi)}{C+\delta(1-\delta)x(Q,\phi)+V}.
    \end{equation}
\end{lemma}

The intuition for the equation in (\ref{eqn-best-response}) is straightforward: Recall from the previous Section \ref{sec-benchmark} that without exclusion the optimal cutoff given $\phi$ is characterized by the equality $W(Q,\phi)= \frac{C}{C+V}$. Here, the cost of rejection is not only the immediate loss $C$, but also the inability to submit in the next period. This latter cost is precisely $\delta x(Q,\phi)-\delta^2 x(Q,\phi) = \delta(1-\delta)x(Q,\phi)$, since instead of having a continuation payoff of $\delta x(Q,\phi)$ the researcher will have a continuation payoff of $\delta^2 x(Q,\phi)$. Unlike in the benchmark case, the cost of rejection here is endogenous and depends both on the competition and on the researcher's own continuation strategy. Note also that, since $x(Q,\phi)>0$ at the optimum, the best response cutoff $Q$ to $\phi$ is higher with exclusion than without, that is, researchers are more conservative and only submit ideas of higher quality.   

By plugging the expression for $x(Q,\phi)$ from (\ref{eqn-steady-state-payoff}) into the optimality conditions in (\ref{eqn-best-response}) we obtain the following equivalent conditions that will be useful below.

\begin{corollary}\label{coro-best-response}
     Assume $v(\phi)> k$. Then $Q\in \RR$ is the best response to $\phi$ if and only if
    \begin{equation}\label{eqn-best-response-alt}
      \frac{V}{C} > \frac{1-F(Q)-Win(Q,\phi)}{Win(Q,\phi)} ~~~~\textit{and}~~~~~  W(Q,\phi)= \frac{C+\delta Win(Q,\phi)}{C+(1+\delta-\delta F(Q))V}.
    \end{equation}
\end{corollary}

\subsection{Steady states}

The next step towards the definition of steady state equilibrium is to identify those quality distributions of submissions that recur every period given a particular strategy profile. Since we have just shown that the best response to any recurring $\phi$ is unique and characterized by a cutoff, we may restrict attention to profiles in which all researchers follow the same cutoff strategy.

Suppose that $Q$ is the common cutoff used by all researchers. Suppose further that in a certain period the fraction of eligible researchers is $\alpha\in [0,1]$. Note that this means that the quality distribution of submissions in this period is $\alpha f^Q$, so in particular the volume of submissions is $\alpha(1-F(Q))$. We distinguish between two possible cases: (i) If $\alpha(1-F(Q))\le k$ then there will be no rejections and hence full eligibility in the next period. (ii) If $\alpha(1-F(Q))> k$ then the volume of rejections is $\alpha(1-F(Q))- k$ and therefore the fraction of eligible researchers in the next period is $1-[\alpha(1-F(Q))- k]$. 

Now, for $\alpha f^Q$ to also be the quality distribution in the next period we need that the fraction of eligible researchers remains $\alpha$. This happens either if $\alpha=1$ and $Q\ge Q^*$ (corresponding to case (i)), or if $\alpha=1-[\alpha(1-F(Q))- k]$ and $Q< Q^*$ (corresponding to case (ii)). The latter is the same as $\alpha=\frac{1+k}{2-F(Q)}$.

To summarize, the steady state quality distribution of submissions induced by the cutoff $Q$ is given by 
\begin{eqnarray}\label{eqn-submissions}
\phi^Q :=
\left\{ \begin{array}{ll}
\frac{1+k}{2-F(Q)}f^Q & \textit{ if } ~~ Q <Q^*,\\
f^Q & \textit{ if } ~~Q \ge Q^*.\\
\end{array} \right.
\end{eqnarray}
We remark that (\ref{eqn-submissions}) naturally extends to the cases $Q=\pm \infty$. Note also that a higher cutoff $Q$ is associated with higher steady state eligibility and with lower steady state submission volume. 


\subsection{Equilibrium}

\begin{definition}\label{def-eq-with-bans}
A cutoff $Q_1\in \RR \cup \{\pm \infty\}$ is a steady state equilibrium with exclusion if $Q_1$ is the best response to $\phi^{Q_1}$ given in (\ref{eqn-submissions}).
\end{definition}

\begin{theorem}\label{thm-steady-state}
    Suppose that $\frac{V}{C}\ge \frac{1-k}{2k}$. Then there exists a steady state equilibrium with exclusion. If $Q_1$ is an equilibrium cutoff then $-\infty<Q_1<Q^*$. Moreover, $Q_1$ is an equilibrium cutoff if and only if
    \begin{equation}\label{eqn-steady-state-eq}
    W(Q_1,\phi^{Q_1})= \frac{(1+k)C+k\delta (2-F(Q_1))V}{(1+k)C+(1+k)(1+\delta -\delta F(Q_1))V}.
    \end{equation}
\end{theorem}


\begin{proof}
    We postpone the proof of existence and first argue the other parts of the theorem. Suppose that $Q_1$ is an equilibrium cutoff. If $Q_1\ge Q^*$ then $v(\phi^{Q_1})\le k$. By part (i) of Lemma \ref{lemma-best-response} the best response to $\phi^{Q_1}$ is to always submit ($Q=-\infty)$, so $Q_1$ is not a best response to $\phi^{Q_1}$, contradiction. Next, suppose that $Q_1=-\infty$. Then $v(\phi^{Q_1})=\frac{1+k}{2}>k$, so again by Lemma \ref{lemma-best-response} $Q_1$ is not the best response to $\phi^{Q_1}$. We conclude that $-\infty<Q_1<Q^*$.

    Next, we show that (\ref{eqn-steady-state-eq}) is necessary and sufficient for $Q_1$ to be an equilibrium. By Corollary \ref{coro-best-response}, $Q$ is the best response to $\phi^Q$ if and only if 
    $$\frac{V}{C} > \frac{1-F(Q)-Win(Q,\phi^Q)}{Win(Q,\phi^Q)} ~~~\textit{and}~~~ W(Q,\phi^{Q})= \frac{C+\delta Win(Q,\phi^{Q})}{C+(1+\delta-\delta F(Q))V}.$$
    We have
    $$Win(Q,\phi^Q) = \int_Q^{\infty} f(q)W(q,\phi^Q)dq = \frac{2-F(Q)}{1+k} \int_{-\infty}^{\infty} \phi^{Q}(q) W(q,\phi^Q)dq = \frac{k(2-F(Q))}{1+k},$$
    where the second equality is from (\ref{eqn-submissions}), and the last equality is by \textbf{A2}. Plugging this into the above best-response conditions and rearranging gives
    $$\frac{V}{C} > \frac{1-k-F(Q)}{k(2-F(Q))} ~~~\textit{and}~~~ W(Q,\phi^{Q})= \frac{(1+k)C+k\delta (2-F(Q))V}{(1+k)C+(1+k)(1+\delta -\delta F(Q))V}.$$
    The first condition trivially holds for any $Q$ by our assumption that $\frac{V}{C}\ge \frac{1-k}{2k}$, and the second condition is precisely (\ref{eqn-steady-state-eq}).
    
    
    Finally, to show existence we follow the footsteps of the proof in Proposition \ref{prop-benchmark}. By the same arguments as in that proof we have $\lim_{Q\downarrow -\infty} W(Q,\phi^Q) = 0$ and $\lim_{Q\uparrow Q^*} W(Q,\phi^Q) = 1$. Also, $W(Q,\phi^Q)$ is continuous in $Q$ on $(-\infty,Q^*)$. The right-hand side of (\ref{eqn-steady-state-eq}) is also continuous in $Q$, and it is easy to see that it is bounded between $\frac{C}{C+(1+\delta)V}$ and $\frac{C+\delta V}{C+(1+\delta k)V}$. This shows that for some $Q_1\in(-\infty,Q^*)$ equality (\ref{eqn-steady-state-eq}) holds, and therefore that a steady state equilibrium with exclusion exists.
\end{proof}

\medskip

Unlike in the benchmark, here we cannot rule out existence of multiple equilibria. First, the left-hand side of (\ref{eqn-steady-state-eq}) need not be monotone in $Q_1$. To see this note that if $Q_1<Q'_1$ then $\phi^{Q_1}$ and $\phi^{Q'_1}$ cannot be pointwise ranked, i.e., $\phi^{Q_1}(q)>\phi^{Q'_1}(q)$ for $Q_1<q<Q'_1$ but $\phi^{Q_1}(q)<\phi^{Q'_1}(q)$ for $q>Q'_1$. Thus, \textbf{A4} can't be used to conclude that a higher $Q_1$ leads to a higher probability of winning, or vice versa. Second, the right-hand side of (\ref{eqn-steady-state-eq}) may be either increasing or decreasing in $Q_1$. Therefore, in general, (\ref{eqn-steady-state-eq}) may hold at multiple cutoffs $Q_1$. However, we do get uniqueness under mild conditions whenever $W$ is an REF (recall subsection \ref{subsec-noisy-review}).


\begin{proposition}\label{prop-unique-ban}
   Suppose that $\frac{V}{C}\ge \frac{1}{k(1-\delta)}$, that $W$ is an REF, and that the ratio $\frac{f(q+d)}{f(q)}$ is decreasing in $q$ for any given $d>0$. Then there is a unique steady state equilibrium with exclusion.
\end{proposition}

The proof shows that under the assumptions of the proposition the function $W(Q,\phi^Q)$ is strictly increasing and that the right-hand side of (\ref{eqn-steady-state-eq}) is decreasing, implying uniqueness. Note that the assumption that $\frac{f(q+d)}{f(q)}$ is decreasing holds for example when $f$ is a normal distribution.

\subsection{Comparison with benchmark}\label{subsec-comparison}

There are two forces that work in opposite directions and determine whether the introduction of an exclusion policy would encourage or discourage self-selection. On the one hand, as was already discussed above, the best response of a researcher to any given $\phi$ has a higher cutoff with exclusion than without due to the bigger loss in case of rejection. On the other hand, with exclusion the competition resulting from researchers choosing a given cutoff is less intense since a certain fraction of the population is ineligible. Indeed, when everyone uses cutoff $Q$, the quality distribution of submissions without exclusion is $f^Q$ while it is $\phi^Q = \frac{1+k}{2-F(Q)}f^Q\le f^Q$ with exclusion (the inequality is strict for $q\ge Q$). By \textbf{A4}, $W(Q,f^Q)\le W(Q,\phi^Q)$ for every $Q$, so the left-hand side of (\ref{eqn-benchmark}) is smaller than that of (\ref{eqn-steady-state-eq}). This may push researchers to submit lower quality ideas. The relative strength of each one of these forces determines whether $Q_0<Q_1$ or vice versa.


While in general the effect of exclusion is ambiguous, the comparison is clear in the interesting case where the benefit from a successful submission $V$ is large. This is formalized in the next result and illustrated in the left panel of Figure \ref{figure-Ban-NoBan}.

\begin{figure}
\centering
\includegraphics[width=.33\textwidth]{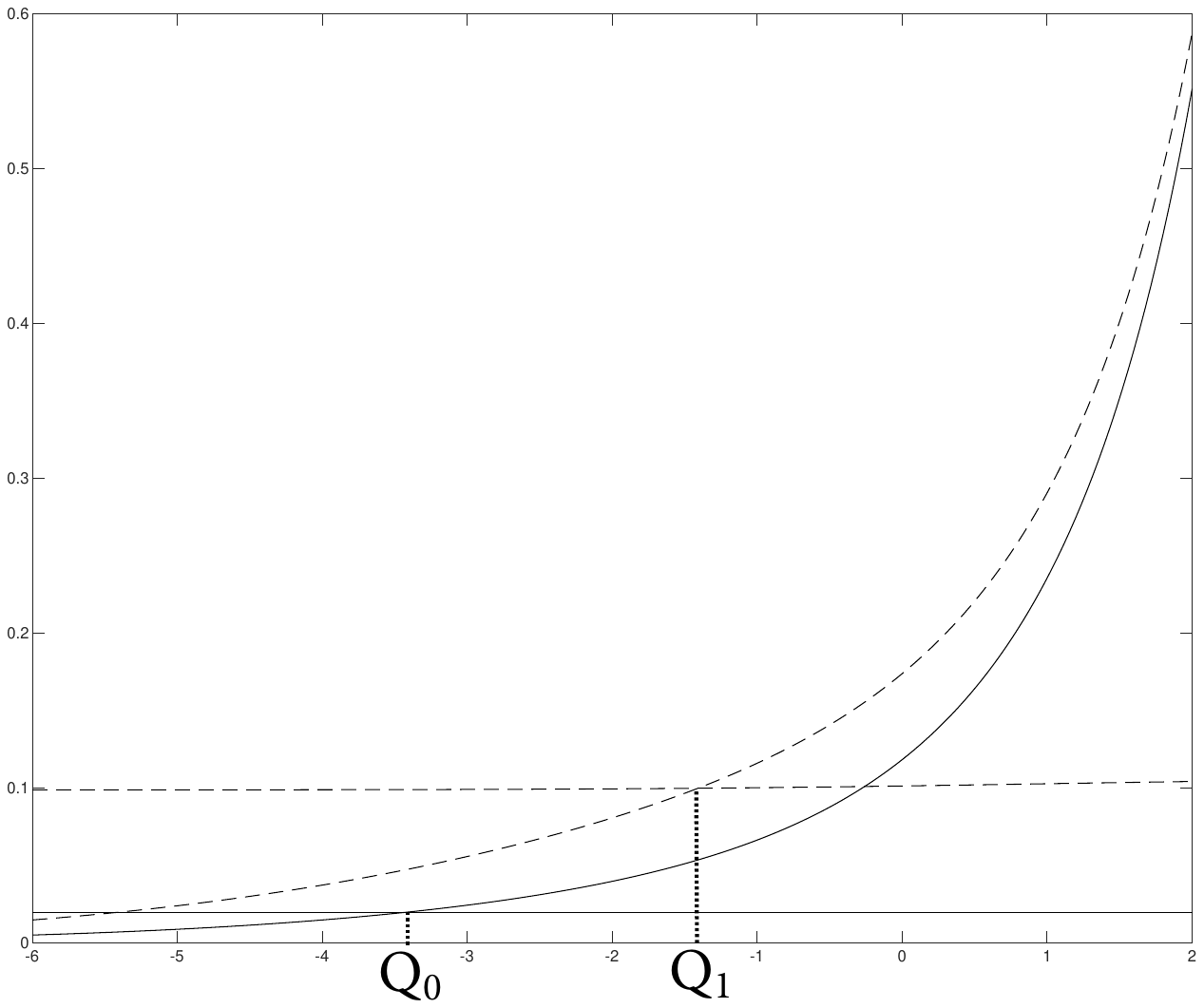}\hfill
\includegraphics[width=.33\textwidth]{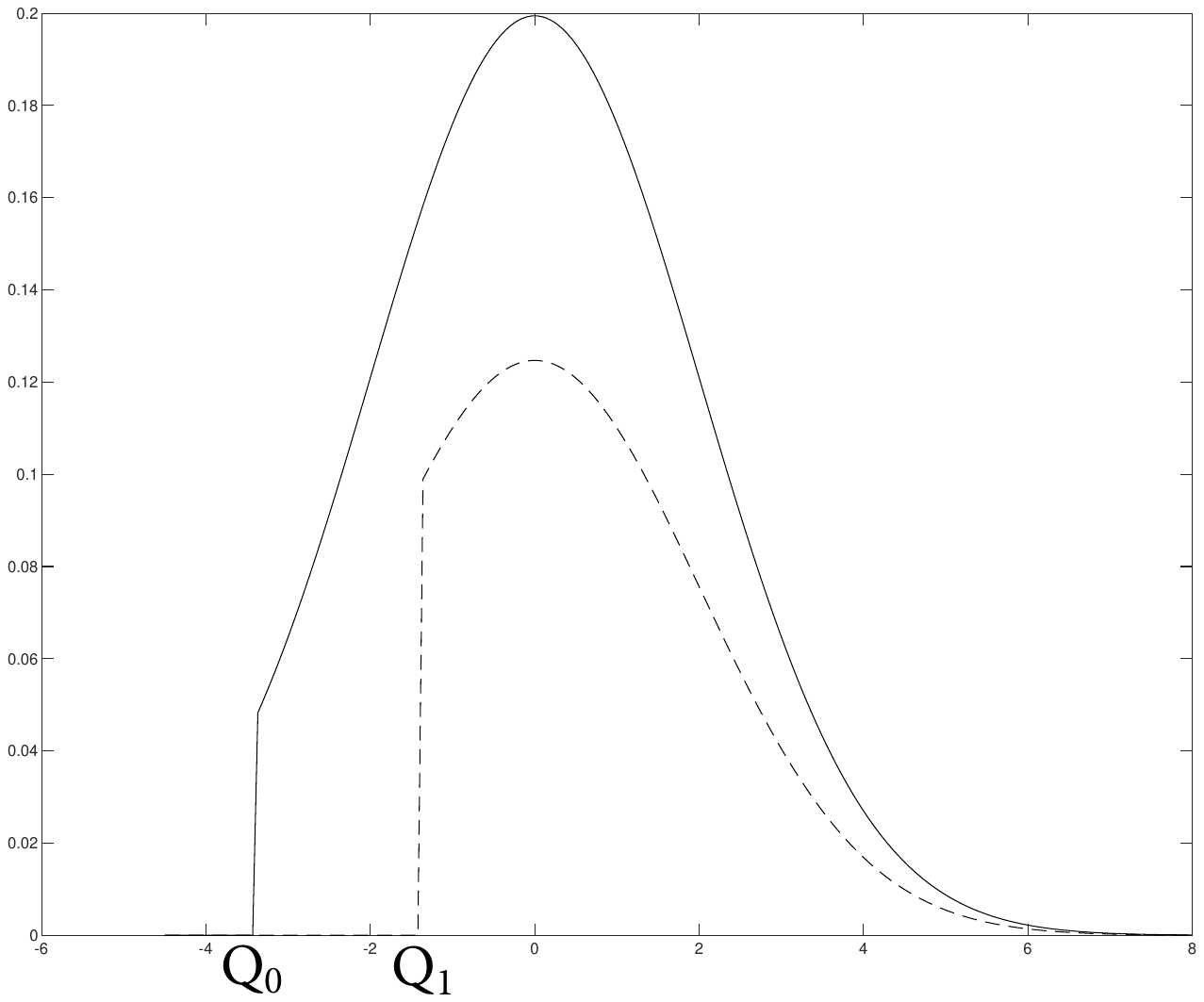}\hfill
\includegraphics[width=.33\textwidth]{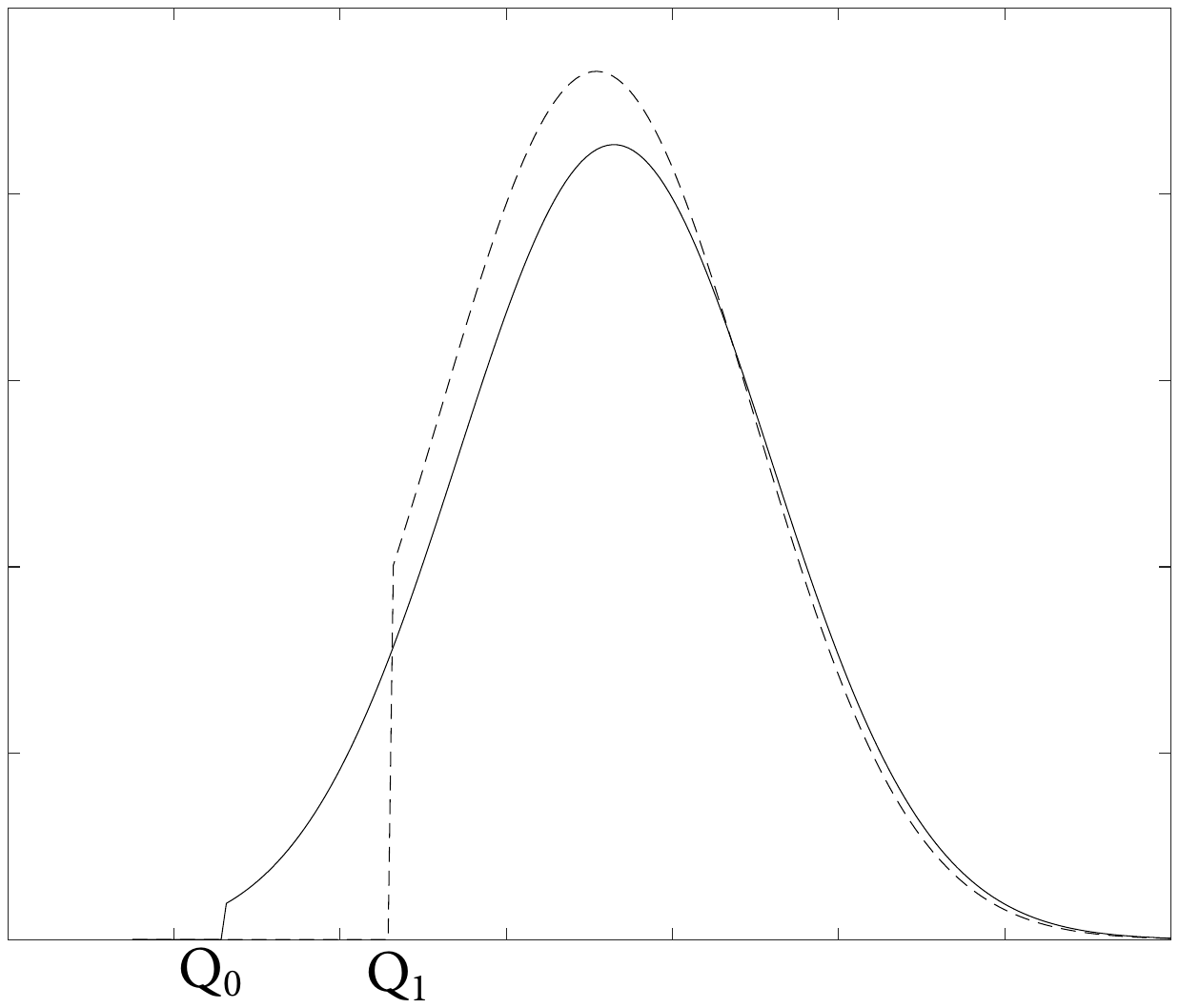}
\caption{Equilibria in the benchmark case (solid) and with exclusion (dashed) in the normal-normal model with parameters $\mu_q=0$, $\sigma^2_q=2$, $\sigma^2_s=5$, $C=1$, $V=50$, $\delta=0.97$, and $k=0.1$. The left panel graphs the two sides of equations (\ref{eqn-benchmark}) (solid) and (\ref{eqn-steady-state-eq}) (dashed), so that the intersection points give the equilibria cutoffs $Q_0$ and $Q_1$. The middle panel shows the induced equilibrium quality distribution of submissions, and the right panel shows the equilibrium quality distribution of funded ideas.}
\label{figure-Ban-NoBan}
\end{figure}

\begin{corollary}\label{coro-comparison-large-V}
    Fix $C,k,\delta,F$ and $W$. For all $V$ large enough the benchmark equilibrium cutoff $Q_0$ is smaller than any steady state equilibrium cutoff with exclusion $Q_1$. 
\end{corollary}
    
\begin{proof}
    Recall from Proposition \ref{prop-benchmark} that the equilibrium cutoff without exclusion is characterized by the equation $W(Q_0,f^{Q_0})=\frac{C}{C+V}$. As $V$ grows to infinity the right-hand side of this equation goes to zero, implying that $\lim_{V\uparrow +\infty} Q_0=-\infty$. 
    
    With exclusion, a steady state equilibrium $Q_1$ is characterized by (\ref{eqn-steady-state-eq}). For any $Q\in (-\infty,Q^*)$ we have 
    $$\lim_{V\uparrow +\infty} \frac{(1+k)C+k\delta (2-F(Q))V}{(1+k)C+(1+k)(1+\delta -\delta F(Q))V} = \frac{k\delta(2-F(Q))}{(1+k)(1+\delta-\delta F(Q))}\ge \frac{k\delta}{1+k\delta}.$$
    On the other hand, we know that $\lim_{Q\downarrow -\infty} W(Q,\phi^{Q})=0$. Let $\underline{Q}$ be such that $W(Q,\phi^{Q})< \frac{1}{2} \cdot \frac{k\delta}{1+k\delta}$ for all $Q<\underline{Q}$. Then it follows from Theorem \ref{thm-steady-state} that for all $V$ large enough any steady state cutoff $Q_1$ must be above $\underline{Q}$. Thus, $Q_1>Q_0$ for all large enough $V$. 
\end{proof}

Note that the equilibrium volume of submissions without exclusion is given by $1-F(Q_0)$, while with exclusion it is $\frac{1+k}{2-F(Q_1)}(1-F(Q_1))$. Since we have just shown that $Q_0<Q_1$ for large $V$, it follows that the volume of submission is lower with exclusion. Indeed, the volume decreases both because low quality applications between the two cutoffs are no longer submitted, and because there are fewer eligible researchers with high quality ideas above $Q_1$. See the middle panel of Figure \ref{figure-Ban-NoBan} for an illustration. Recall that the lower volume also implies that researchers' welfare is higher with exclusion.

The right panel of Figure \ref{figure-Ban-NoBan} compares the resulting quality distributions of winners. Since $Q_0<Q_1$ exclusion eliminates funding for qualities $q\in (Q_0,Q_1)$, simply because these are no longer being submitted. However, exclusion reduces the share of the budget that is allocated to top-quality applications since the equilibrium winning probabilities $W(q,f^{Q_0})$ and $W(q,\phi^{Q_1})$ both converge to one as $q$ increases to infinity (by \textbf{A3}), and there are more top-quality submissions without exclusion. Overall, for large $V$ exclusion tends to shift budget from the two extremes to the middle range of qualities; Proposition \ref{prop-noisy-review-comparison} below shows that this holds for a broad class of exclusion policies. Depending on the model parameters, this shift often results in an increase of the mean quality (as in Figure \ref{figure-Ban-NoBan}), and is beneficial whenever the principal is mostly concerned about the possibility of wasting resources on low quality ideas. On the other hand, if the principal tries to maximize the mean of a sufficiently convex objective, then our model suggests that exclusion is counterproductive.

\section{General exclusion policies}\label{sec-generalized}


\subsection{Other exclusion criteria}\label{subsec-other-criteria}

In the policy considered thus far exclusion was triggered by rejection of a researcher's proposal. When the review process of applications generates a signal of their quality, as is the case with REFs defined in subsection \ref{subsec-noisy-review}, the designer can use different cutoffs for current period funding than for next period eligibility. We consider here the effects that such generalized exclusion criteria have on the steady state equilibria of the induced game.

More formally, suppose as in subsection \ref{subsec-noisy-review} that a proposal of quality $q$ generates a signal $s$ whose (conditional) cdf is given by $G(s|q)=G(s-q)$, where $G$ admits a continuous and everywhere positive density $g$. The cutoff signal for funding given submission distribution $\phi$, denoted $\bar s(\phi)$, is determined by the market clearing condition (\ref{eqn-signal-cutoff}). Fix some $\bar s\in \RR \cup \{\pm \infty\}$ and consider the policy according to which a researcher is excluded in the next period if and only if their proposal in the current period generated a signal below $\bar s$. The special case $\bar s = -\infty$ is the benchmark of Section \ref{sec-benchmark}, while the other extreme $\bar s=+\infty$ means that a researcher who applied today cannot participate tomorrow, regardless of the outcome.\footnote{Note that we assume here that the principal commits to an exclusion cutoff that is independent of the submitted quality distribution $\phi$, and hence that the volume of excluded researchers would change depending on $\phi$ (unlike the funding cutoff which adjusts to induce mass $k$ of winners). However, if $Q_1$ is an equilibrium cutoff with exclusion as in Section \ref{sec-steady-states-conditional}, then setting $\bar s = \bar s(\phi^{Q_1})$ means that $Q_1$ is still an equilibrium and the resulting outcome is the same as in that section; other equilibria may be affected by this modification.} 

It will be convenient to denote by $Ban(Q,\bar s) = \int_Q^{\infty} f(q)G(\bar s-q)dq$ the ex-ante probability that an eligible researcher will be excluded in the next period given that they use entry cutoff $Q$ and given exclusion cutoff $\bar s$. Using this notation, and similar to the derivation of (\ref{eqn-steady-state-payoff2}), the total expected payoff of a researcher that uses cutoff $Q$ and faces the same competition $\phi$ every period (and given $\bar s$) is then
$$x(Q,\phi,\bar s) := \frac{Win(Q,\phi)V-(1-F(Q)-Win(Q,\phi))C}{(1-\delta)(1+\delta Ban(Q,\bar s))}.$$

Now, if $\alpha$ is the share of eligible researchers in a given period, and if all researchers use entry cutoff $Q$, then the volume of rejections in this period is given by $\alpha Ban(Q,\bar s)$. Thus, the share of eligible researchers next period is $1-\alpha Ban(Q,\bar s)$. A steady state share of eligibility is therefore given by $\alpha = \frac{1}{1+Ban(Q,\bar s)}$. This motivates the following definition.

\begin{definition}\label{def-eq-noisy-review}
Fix an exclusion cutoff $\bar s$. A cutoff $Q\in \RR \cup \{\pm \infty\}$ is a steady state equilibrium given $\bar s$ if $Q$ is a best response to $\phi=\frac{1}{1+Ban(Q,\bar s)}f^Q$, i.e.\ if 
$$Q\in \argmax_{Q'} x\left(Q',\frac{1}{1+Ban(Q,\bar s)}f^Q,\bar s\right).$$
\end{definition}

The next result generalizes Corollary \ref{coro-comparison-large-V}, showing that any non-trivial exclusion cutoff induces more self-selection than in the benchmark without exclusion whenever the benefit of winning $V$ is sufficiently large. 

\begin{proposition}\label{prop-noisy-review-large-V}
Fix $C,k,\delta,F,W$ and an exclusion cutoff $\bar s\in (-\infty,+\infty]$ such that $k<\frac{1}{1+Ban(-\infty,\bar s)}$.\footnote{Note that this condition holds whenever $k<0.5$. If $k\ge\frac{1}{1+Ban(-\infty,\bar s)}$ then the unique steady state equilibrium cutoff given $\bar s$ is $Q=-\infty$, i.e., every eligible researcher applies, and every application wins.} For all $V$ large enough the benchmark equilibrium cutoff $Q_0$ is smaller than any steady state equilibrium cutoff $Q$ given $\bar s$. 
\end{proposition}

Our last result in this section compares the quality distribution of winners between the benchmark equilibrium and some equilibrium with exclusion.  

\begin{proposition}\label{prop-noisy-review-comparison}
Suppose that the signal distribution $G$ has increasing hazard rate. Fix an exclusion cutoff $\bar s\in (-\infty,+\infty]$ and suppose that $Q$ is an equilibrium cutoff given $\bar s$. Denote by $h$ the resulting equilibrium quality distribution of winners and by $h_0$ the quality distribution of winners in the benchmark equilibrium.\\
(i) If $Q>Q_0$ then there is $\bar Q>Q$ such that $h(q)\ge h_0(q)$ for $q\in[Q,\bar Q]$ and $h(q)\le h_0(q)$ otherwise.\\
(ii) If $Q\le Q_0$ then $h_0$ first-order stochastically dominates $h$.
\end{proposition}

When the designer sets an exclusion cutoff $\bar s$ their goal is to encourage more self-selection. Part (i) of Proposition \ref{prop-noisy-review-comparison} concerns the case where this attempt is successful, raising the entry cutoff from $Q_0$ to $Q$. In this case, exclusion results in elimination of funding for qualities between $Q_0$ and $Q$ (since these are no longer being submitted), but at the same time some reduction in funding of top-quality applications (above $\bar Q$); thus, funding is shifted from the two extremes towards the middle of the range of qualities. See the right panel of Figure \ref{figure-Ban-NoBan} for an illustration. On the other hand, exclusion may result in a lower entry cutoff $Q\le Q_0$ due to the reduced competition that an eligible researcher faces. Part (ii) of the proposition argues that in this case the quality distributions of winners can be unambiguously ranked, with exclusion leading to a worse outcome. Note that, by Proposition \ref{prop-noisy-review-large-V}, whenever $V$ is large enough we have that $Q>Q_0$ and hence the relevant case is (i).


\subsection{Multi-period exclusion}\label{subsec-multi-period}

Since the goal of exclusion is to encourage self-selection, one may wonder whether making rejection more costly by increasing the length of ineligibility provides the right incentives. We study policies according to which rejection implies a ban of $t\ge 1$ periods. The analysis repeats the same steps as in the $t=1$ case analyzed above so we skip some of the details.

The total expected payoff of an eligible researcher that uses cutoff $Q$ every period (when eligible), given that the quality distribution per-period is $\phi$, is given by
$$x_t(Q,\phi) =\frac{Win(Q,\phi)V-(1-F(Q)-Win(Q,\phi))C}{(1-\delta)[1 + \delta(1-F(Q)-Win(Q,\phi))(1+\delta+\ldots+\delta^{t-1})]}.$$
Notice that $t$ only shows up once in the denominator and that $x_t(Q,\phi)$ is decreasing in $t$. Just as in Lemma \ref{lemma-best-response}, if $v(\phi)\le k$ then $Q=-\infty$ (always submit) is the unique best response, and if $v(\phi)>k$ then the unique best response is characterized by $x_t(Q,\phi)> 0$ and
    \begin{equation}\label{eqn-best-response-T}
        W(Q,\phi)= \frac{C+\delta(1-\delta^t)x_t(Q,\phi)}{C+\delta(1-\delta^t)x_t(Q,\phi)+V}.
    \end{equation}
It is not hard to see that the product $(1-\delta^t)x_t(Q,\phi)$ increases in $t$, from which we get the following.

\begin{lemma}\label{lemma-increasing-T}
    If $v(\phi)> k$ then the best-response cutoff to $\phi$ increases in the length of exclusion $t$.
\end{lemma}

Thus, for fixed competition $\phi$, the longer the ban the more certain of wining a researcher needs to be in order to submit. Of course, changing $t$ would also change the steady state competition that a researcher faces. Namely, suppose that all researchers use cutoff $Q$, and that $\alpha$ is the share of eligible researchers in some period. The volume of submissions is then $\alpha(1-F(Q))$. In the relevant case where this volume is greater than $k$ there are $\alpha(1-F(Q))-k$ rejections, and these researchers will be ineligible for $t$ periods. If this happens every period then the per-period fraction of eligible researchers is $1-t[\alpha(1-F(Q))-k]$. For a steady state we thus need $\alpha = 1-t[\alpha(1-F(Q))-k]$, or $\alpha=\frac{1+tk}{1+t(1-F(Q))}$. The steady state quality distribution of submissions induced by the cutoff $Q$ is therefore  
\begin{eqnarray}\label{eqn-submissions-T}
\phi_t^Q:=
\left\{ \begin{array}{ll}
\frac{1+tk}{1+t(1-F(Q))}f^Q & \textit{ if } ~~ Q <Q^*,\\
f^Q & \textit{ if } ~~Q \ge Q^*.\\
\end{array} \right.
\end{eqnarray}

\begin{definition}\label{def-eq-with-bans-T}
A cutoff $Q_t\in \RR \cup \{\pm \infty\}$ is a steady state equilibrium with $t$-period exclusion if $Q_t$ is the best response to $\phi_t^{Q_t}$ given in (\ref{eqn-submissions-T}).
\end{definition}

The following generalizes Theorem \ref{thm-steady-state}. The proof is similar and therefore omitted.

\begin{theorem}\label{thm-steady-state-T}
    Suppose that $\frac{V}{C}\ge \frac{1-k}{(t+1)k}$. Then there exists a steady state equilibrium with $t$-period exclusion. If $Q_t$ is an equilibrium cutoff then $-\infty<Q_t<Q^*$. Moreover, $Q_t$ is an equilibrium cutoff if and only if
    \begin{equation}\label{eqn-steady-state-eq-T}
    W(Q_t,\phi^{Q_t}_t)= \frac{(1+tk)C+k\delta(1+\delta+\ldots+\delta^{t-1}) (1+t(1-F(Q_t)))V}{(1+tk)C+(1+tk)[1+\delta(1+\delta+\ldots+\delta^{t-1})(1-F(Q_t))]V}.
    \end{equation}     
\end{theorem}

We know from Lemma \ref{lemma-increasing-T} above that the best response cutoff for a given $\phi$ increases in $t$. On the other hand, the share of eligible researchers given cutoff $Q$ is $\frac{1+tk}{1+t(1-F(Q))}$, which decreases in $t$. Thus, similar to the comparison between the benchmark and one-period exclusion of the previous section, when $t$ increases the corresponding equilibrium cutoff may either increase or decrease. Figure \ref{figure-Multi-Period} shows an example in which the sequence of cutoffs $\{Q_t\}_t$ is non-monotonic.


\begin{figure}
\centering
\includegraphics[width=.5\textwidth]{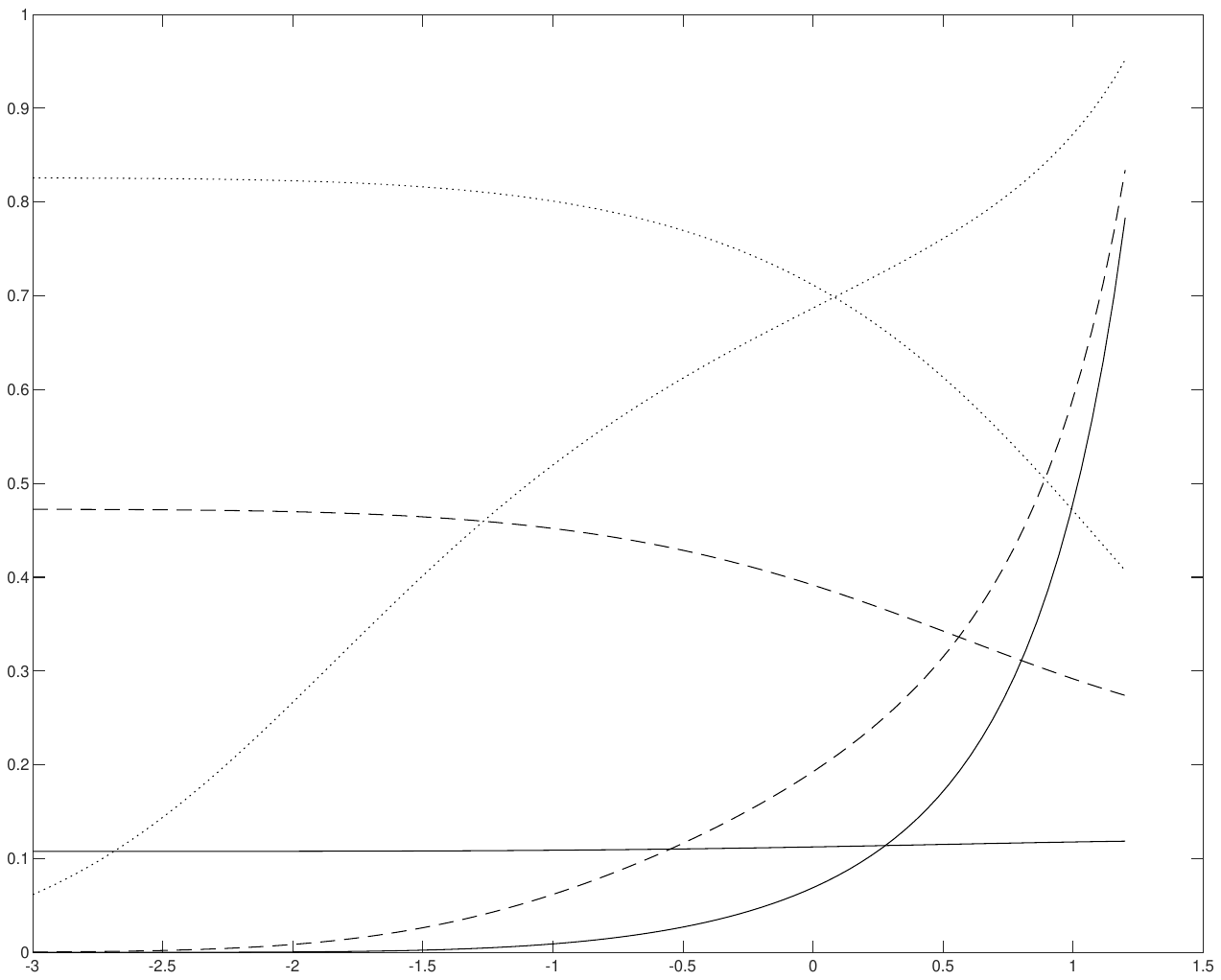}
\caption{Equilibria cutoffs for various lengths of exclusion in the normal-normal model with parameters $\mu_q=0$, $\sigma^2_q=1$, $\sigma^2_s=1$, $C=1$, $V=20$, $\delta=0.85$, and $k=0.1$. The solid lines are for $t=1$, the dashed lines for $t=5$, and the dotted lines for $t=50$. For each of these cases the lines depict both sides of equation (\ref{eqn-steady-state-eq-T}), so that the intersection point is the equilibrium cutoff. Notice that the cutoffs satisfy $Q_{50}<Q_1<Q_5$.}
\label{figure-Multi-Period}
\end{figure}

A natural question is what happens to the steady state when $t$ becomes large. A naive conjecture may be that in this case researchers will wait until they have an extremely high-quality idea to guarantee winning and avoid being struck out for many periods. It turns out however that exactly the opposite is true, as stated in the following proposition.

\begin{proposition}\label{prop-T-large}
    Fix $C,V,k,\delta,F$ and $W$. If for each $t\ge 1$ $Q_t$ is a steady state equilibrium cutoff with $t$-period exclusion then $\lim_t Q_t=-\infty$. Moreover, (i) the share of eligible researchers converges to $k$, (ii) the quality distribution of submissions $\phi^{Q_t}_t$ converges to $kf$, and (iii) the quality distribution of funded ideas $\phi^{Q_t}_t(\cdot) W(\cdot,\phi^{Q_t}_t)$ converges to $kf$.   
\end{proposition}

Thus, long bans induce the worst possible outcome for the principal: In equilibrium there is virtually no self-selection, almost all eligible researchers apply, and almost everyone who applies wins. The intuition for this result is that the steady state share of eligible researchers uniformly (in $Q$) converges to $k$ as $t$ increases, implying that for any quality $q$ the probability of winning is close to one. Thus, for not applying to be a best response it must be that the cost of rejection grows unbounded as $t$ increases. But even with arbitrarily long bans the cost of rejection is bounded above by $\frac{V}{1-\delta}$, the payoff from winning in every future period. 


\section{Heterogeneous researchers}\label{sec-heterogenous}


Up to now, researchers in our model were identical in terms of the distribution of quality of ideas they generate each period. In this section we consider an extension where researchers have persistent types. Our goal is twofold: The first is to illustrate that the analysis carries over rather easily to this more realistic environment. The second is to show that with heterogeneity exclusion policies generate interesting differential effects across researchers.   

Suppose that there are two types of researchers, High ($H$) and Low ($L$). The fractions of types in the population are $\lambda_H$ and $\lambda_L$ ($\lambda_H+\lambda_L=1$) for $H$ and $L$ types, respectively. The per-period quality distribution generated by type $i$ researchers is $F_i$, with associated positive and continuous density $f_i$ ($i=H,L$). In the entire population the density is therefore $f:=\lambda_Hf_H+\lambda_Lf_L$. All other parts of the model are the same as before. In particular, the reviewing process is `type blind', meaning that conditional on quality $q$ the probability of winning is independent of the type and thus given by $W(q,\phi)$.\footnote{This assumption is more appropriate in cases where reviewers are unaware of the identity of the submitter.} 

For each type $i\in \{H,L\}$ the set of possible quality distributions of submissions is $\Gamma_{\lambda_i f_i}$, and we denote by $\phi_i$ a typical quality distribution submitted by type $i$ researchers. The sum $\phi=\phi_H+\phi_L$ is then the total quality distribution. We also let 
$$k_i(\phi_H,\phi_L) = \int \phi_i(q) W(q,\phi_H+\phi_L)dq$$
be the volume of winning applications of type $i$ researchers. Note that $k_L(\phi_H,\phi_L)+k_H(\phi_H,\phi_L)=k$ whenever $v(\phi_L+\phi_H)>k$ by \textbf{A2}. 

Consider first the benchmark without exclusion. Since there is no interdependence between different periods, the optimal action of a researcher conditional on their current period quality $q$ is independent of their type. Thus, the analysis of this case is exactly the same as in Section \ref{sec-benchmark} above, and the equilibrium is characterized by a single cutoff $Q_0$ that satisfies $W(Q_0,f^{Q_0})=\frac{C}{C+V}$ as in Proposition \ref{prop-benchmark}. Both types of researchers apply if and only if their current period quality is at least $Q_0$.

Suppose now that an exclusion policy as in Section \ref{sec-steady-states-conditional} is introduced. As in that section, let 
\begin{equation}\label{eqn-steady-state-payoff-types}
x_i(Q,\phi) :=\frac{Win_i(Q,\phi)V-(1-F_i(Q)-Win_i(Q,\phi))C}{(1-\delta)[1 + \delta(1-F_i(Q)-Win_i(Q,\phi))]}.
\end{equation}
be the payoff of a researcher of type $i$ that follows cutoff $Q$ and faces competition $\phi$ every period, where $Win_i(Q,\phi) := \int_Q^{\infty} f_i(q)W(q,\phi)dq$ is the winning probability of the researcher. 

To define a steady state equilibrium, we need to describe the steady state share of eligible researchers of each type associated with a given strategy profile. Suppose that all type $i$ researchers follow the strategy defined by a cutoff $Q^i$ ($i=H,L$). Similar to (\ref{eqn-submissions}), there are two cases to consider: If $\lambda_H(1-F_H(Q^H))+\lambda_L(1-F_L(Q^L))\le k$ then there would be no rejections and the shares of eligibility in each period are $(\lambda_H,\lambda_L)$. If on the other hand $\lambda_H(1-F_H(Q^H))+\lambda_L(1-F_L(Q^L))> k$ then steady state eligibility shares $(\alpha_H,\alpha_L)$ must satisfy
\begin{equation}\label{eqn-fixed-point-types}
    \alpha_i=\lambda_i - \left[\alpha_i(1-F_i(Q^i)- k_i\left(\alpha_H f_H^{Q^H},~ \alpha_L f_L^{Q^L}\right) \right].
\end{equation}

\begin{definition}\label{def-fixed-point-types}
    We say that $(\alpha_H,\alpha_L)\in [0,\lambda_H] \times [0,\lambda_L]$ are steady state eligibility shares for cutoffs $(Q^H,Q^L)$ if\\
    (i) $\alpha_i=\lambda_i$ for $i=H,L$ when $\lambda_H(1-F_H(Q^H))+\lambda_L(1-F_L(Q^L))\le k$; \\
    (ii) $\alpha_i$ satisfies (\ref{eqn-fixed-point-types}) for $i=H,L$ when $\lambda_H(1-F_H(Q^H))+\lambda_L(1-F_L(Q^L))> k$.
\end{definition}



We note that it follows from a standard fixed-point argument that steady state eligibility shares $(\alpha_H,\alpha_L)$ exist for every pair of cutoffs $(Q^H,Q^L)$. We can now define steady state equilibrium with exclusion.

\begin{definition}\label{def-equilibrium-types}
    The pair of cutoffs $(Q_1^H,Q_1^L)$ is a steady state equilibrium with exclusion if there exist steady state eligibility shares $(\alpha_H,\alpha_L)$ for $(Q_1^H,Q_1^L)$ such that $Q_1^i$ is type's $i$ best response to $\phi=\alpha_H f_H^{Q_1^H} + \alpha_L f_L^{Q_1^L}$ for both $i=H,L$. 
\end{definition}



As in the case of homogeneous population, when the benefit of winning $V$ is sufficiently large, the benchmark equilibrium cutoff $Q_0$ is smaller than any equilibrium cutoffs $Q_1^H$ and $Q_1^L$ with exclusion. We omit the proof of this result as it follows the same logic as in the proofs of Corollary \ref{coro-comparison-large-V} and Proposition \ref{prop-noisy-review-large-V}. 

Once exclusion policy is introduced, researchers of different types face different incentives when deciding whether to apply, even conditional on having the same quality $q$. Indeed, if $H$ types tend to produce ideas of higher quality than $L$ types, then the option value of participating in the next period is higher for $H$ types. Thus, $H$ types face a larger downside risk and therefore require a higher winning probability in order to apply. The next proposition formalizes this intuition.

\begin{proposition}\label{prop-compare-types}
If $F_H$ first-order stochastically dominates $F_L$, and if $(Q_1^H,Q_1^L)$ is a steady state equilibrium with exclusion, then $Q_1^H \ge Q_1^L$. Moreover, if $F_H(q)<F_L(q)$ for all $q$, then $Q_1^H > Q_1^L$.
\end{proposition}

As a final comment we note that the definitions and results of this section naturally extend to populations with any finite number of types.


\section{Discussion}\label{sec-discussion}

We presented a stylized model of repeated contests in which both the principal and the agents gain from higher levels of self-selection, and studied whether and when temporary exclusion policies help achieving this goal. We end with a short discussion of several important issues that are absent from the paper and may be interesting directions for future work.

Perhaps the most obvious extension is to endogenize the quality distribution in the population by letting quality be a (noisy) function of effort, as is standard in the contests literature. When an exclusion policy is put in place, incentives for exerting effort will be affected, and the quality distribution in the population will be altered. Consider the simplest version with a binary effort choice: Each researcher can choose either high effort at cost $c>0$ or low effort at cost zero; quality is distributed $F_H$ or $F_L$ for a researcher that exerts high or low effort, respectively; once quality is realized, a researcher decides whether to submit or not. Note that if $\lambda_H$ is the fraction of researchers in the population that chooses high effort, then the submission subgame is as in Section \ref{sec-heterogenous} above. However, here this fraction will be determined by an indifference condition, namely, that the payoff gap between high and low effort is equal to the cost of effort $c$. In particular, $\lambda_H$ will typically vary depending on whether or not an exclusion policy is present. We believe that the framework developed in this paper is suitable for analyzing such moral hazard issues, but do not pursue it here. 

As discussed above, previous literature has suggested that `quota mechanisms' which limit the number of times an action can be taken perform well in environments similar to ours. To illustrate, suppose that periods are divided into blocks of size $T$ and each agent can apply at most $kT$ times in each block. Intuitively, if there was no discounting, and if $T$ is large, then the strategy profile in which each agent applies only when their quality is above the first-best cutoff $Q^*$ (until they run out of `budget') is an approximate equilibrium that virtually implements the first-best.\footnote{To make this statement precise one would need to work with a finite number of agents as in \citet{jackson2007overcoming}. See also \citet{hortala2010inefficiencies} and \citet{ball2023linking} for related results when $T$ is away from the limit.} The practical relevance of such mechanisms with large $T$ is however questionable. Furthermore, with discounting the quota of an agent would need to be adjusted over time as in \citet{frankel2016discounted}, and analyzing the equilibria of the induced repeated game is likely challenging. Our approach focusing on steady states makes the model tractable. In any case, it would be interesting to more carefully compare the two approaches.     


We intentionally did not specify an objective function for the principal. Our goal was to keep the model as general as possible to highlight the key trade-offs that exclusion policies introduce to the equilibrium supply of ideas, and to consider policies that are simple and easy to implement. Presumably, however, the principal has preferences over quality distributions of submissions (which induce quality distributions of winners), and would try to design submission policies to maximize this preference. The tractability of such mechanism design exercise would depend on the objective function, the space of allowed mechanisms, and on the specification of the `noise function' $W$. In particular, eligibility can in principle depend on the history of play in a complex way, so to make progress one would likely need to focus on some parametric family of policies similar to the multi-period bans of subsection \ref{subsec-multi-period}.  

Finally, there are many ways in which our model can be generalized to better capture features of real-world contests. Our assumption that the environment is stationary is clearly restrictive as it precludes the possibility that researchers produce on average higher quality ideas over time due to learning, or that they can resubmit rejected ideas from previous rounds; we did not consider heterogeneous awards, and moreover assumed that the benefit of winning $V$ is independent of the idea quality $q$; all the agents in our model discount future payoffs at the same rate, ruling out differences based on career stage; and we have assumed away the possibility that referees' evaluations would change once exclusion policy is introduced. These interesting extensions are beyond the scope of this paper.








\bibliographystyle{plainnat}
\bibliography{temporary_bans}

\appendix

\section{Missing proofs}\label{sec-proofs}

\noindent \textbf{Proof of Proposition \ref{prop-benchmark}}

First, any $Q\ge Q^*$ cannot be an equilibrium cutoff since then $v(f^Q)=1-F(Q)\le k$, implying that the best response is to always submit. Second, $Q=-\infty$ (always submit) is not an equilibrium since by \textbf{A3} $\lim_{q\downarrow -\infty} W(q,f)=0$, so all sufficiently low qualities prefer not to submit. 

Now, let $l:(-\infty,Q^*)\to [0,1]$ be defined by $l(Q)=W(Q,f^Q)$. If $Q<Q'<Q^*$ then 
$$l(Q)=W(Q,f^Q) < W(Q',f^Q) \le W(Q',f^{Q'}) = l(Q'),$$
where the strict inequality is by \textbf{A3} and the weak inequality is by \textbf{A4}. Thus, $l$ is strictly increasing. Further, since the mapping $Q\to f^Q$ is continuous, it immediately follows from \textbf{A1} that $l$ is continuous. Also, 
$$\lim_{Q\uparrow Q^*} l(Q) = W(Q^*,f^{Q^*}) = 1,$$
where the first equality is by \textbf{A1} and the second by \textbf{A2}. Next, we claim that $\lim_{Q\downarrow -\infty} l(Q) = 0$. Indeed, for a given $\epsilon>0$ let $Q'$ be small enough so that $W(Q',f)<\epsilon$ (by \textbf{A3}), and then pick $Q''<Q'$ small enough so that $|W(Q', f^{Q''})- W(Q',f)|<\epsilon$ (by \textbf{A1}). Then 
$$l(Q'') = W(Q'',f^{Q''}) < W(Q',f^{Q''}) \le  |W(Q', f^{Q''})- W(Q',f)| + W(Q',f) < 2\epsilon,$$
where the first inequality is by $\textbf{A3}$, and the last is by the choice of $Q'$ and $Q''$. Since $l$ was proved to be increasing it follows that $l(Q)<2\epsilon$ for all $Q< Q''$.    

Therefore, there exists a unique $Q_0$ such that $l(Q_0)=\frac{C}{C+V}$. Clearly, $Q_0$ is a best response to $f^{Q_0}$ if and only if $l(Q_0)=\frac{C}{C+V}$.

\bigskip

\noindent \textbf{Proof of Lemma \ref{lemma-best-response}}

\noindent (i) If $v(\phi)\le k$ then, by \textbf{A2}, $W(q,\phi)=1$ for every $q$, implying that submitting every period results in the maximal possible payoff of $\frac{V}{1-\delta}$.

\noindent (ii) Fix $\phi$ with $v(\phi)> k$. Suppose that the conditions (\ref{eqn-best-response}) hold. We can rewrite the equality in (\ref{eqn-best-response}) as 
    $$W(Q,\phi)(V+\delta x(Q,\phi)) + (1-W(Q,\phi))(-C+\delta^2 x(Q,\phi))= \delta x(Q,\phi).$$
    The left-hand side is the total expected payoff of a researcher that has type $Q$ in the current period, chooses to submit, and follows the threshold strategy $Q$ in every future period; the right-hand side is the expected payoff of this researcher if they choose not to submit. Now, since $x(Q,\phi)> 0$ we have $V+\delta x(Q,\phi) > -C+\delta^2 x(Q,\phi)$, implying that the left-hand side is strictly increasing in $W$. Thus, by \textbf{A3}, types above $Q$ strictly prefer to submit and types below $Q$ strictly prefer not to submit. In other words, if the researcher uses the cutoff strategy $Q$ in the future, then it is optimal to use it also in the current period. By \citet[Theorem 1]{blackwell1962discrete} $Q$ is a best response to $\phi$.

    Conversely, suppose that $Q\in \RR$ is a best response to $\phi$. First, we must have $x(Q,\phi)> 0$ since by \textbf{A3} $W(q,\phi)>\frac{C}{C+V}$ for all large enough $q$, so the researcher can guarantee a positive expected payoff if they set a large enough cutoff. Second, it follows from \citet[Theorem 2]{blackwell1962discrete} that $x(Q,\phi)$ is (weakly) higher than the payoff of a researcher that follows some other cutoff $Q'$ in the current period and then switches to cutoff $Q$ in all future periods. Thus, for all $Q'<Q$ we have 
    \begin{eqnarray*}
    F(Q)(0+\delta x) &+& \int_Q^{\infty} f(q)\left[W(q,\phi)(V+\delta x) + (1-W(q,\phi))(-C+\delta^2 x) \right] dq \ge \\
    & & F(Q')(0+\delta x) + \int_{Q'}^{\infty} f(q)\left[W(q,\phi)(V+\delta x) + (1-W(q,\phi))(-C+\delta^2 x) \right] dq,
    \end{eqnarray*}
    where for simplicity we write $x$ instead of $x(Q,\phi)$. Rearranging gives
    $$\delta x \ge \frac{1}{F(Q)-F(Q')} \int_{Q'}^Q  f(q)\left[W(q,\phi)(V+\delta x)+(1-W(q,\phi))(-C+\delta^2 x) \right]dq.$$
    Taking limit as $Q'\uparrow Q$ results in 
    $$\delta x \ge \left[W(Q,\phi)(V+\delta x)+(1-W(Q,\phi))(-C+\delta^2 x) \right],$$
    which after simple algebra becomes 
    $$W(Q,\phi)\le \frac{C+\delta(1-\delta)x}{C+\delta(1-\delta)x+V}.$$
    Repeating the argument for $Q'>Q$ gives the reverse inequality, proving the equality in (\ref{eqn-best-response}).  

    For existence, consider the two sides of the equation in (\ref{eqn-best-response}) for $Q$'s larger than the point at which $W(Q,\phi)=\frac{C}{C+V}$. The left-hand side is equal to $\frac{C}{C+V}$ at this point and is strictly increasing and continuous thereafter. The right-hand side value at this point is strictly larger than $\frac{C}{C+V}$ ($x(Q,\phi)$ is positive at any $Q$ in this interval since the researcher has a positive expected payoff in every period they submit), it is continuous, and it converges to $\frac{C}{C+V}$ as $Q\uparrow +\infty$. Thus, there exists $Q$ in this interval where (\ref{eqn-best-response}) holds.
            
        
    For uniqueness, suppose that $Q,Q'\in \RR$ are both optimal. Then from the equality in (\ref{eqn-best-response}) we get 
    $$W(Q,\phi)= \frac{C+\delta(1-\delta)x(Q,\phi)}{C+V+\delta(1-\delta)x(Q,\phi)} = \frac{C+\delta(1-\delta)x(Q',\phi)}{C+V+\delta(1-\delta)x(Q',\phi)} = W(Q',\phi),$$
    so by strict monotonicity of $W(\cdot,\phi)$ in \textbf{A3} we have $Q=Q'$.

\bigskip

\noindent \textbf{Proof of Proposition \ref{prop-unique-ban}}

     The right-hand side of (\ref{eqn-steady-state-eq}) is the ratio of two linear functions of $F(Q_1)$ and is therefore monotone. When $Q_1 \to -\infty$ this ratio converges to $\frac{(1+k)C+2k\delta V}{(1+k)C+(1+k)(1+\delta)V}$, and when $Q_1\to Q^*$ this ratio converges to $\frac{C+k\delta V}{C+(1+\delta k)V}$. The condition $\frac{V}{C}\ge \frac{1}{k(1-\delta)}$ implies that the former is larger than the latter, and hence that the right-hand side of (\ref{eqn-steady-state-eq}) decreases in $Q_1$. 
    

    Next, we argue that the left-hand side of (\ref{eqn-steady-state-eq}), $W(Q,\phi^Q)$, is increasing in $Q$. To see this, recall that $W(Q,\phi^Q)=1-G(\bar s(\phi^Q)-Q)$, where $\bar s(\phi^Q)$ is determined by $\int \phi^Q(q)[1-G(\bar s(\phi^Q)-q)]dq=k$. Plugging in $\phi^Q=\frac{1+k}{2-F(Q)}f^Q$ this last equality becomes
    \begin{equation}\label{eqn-normal-ban}
        \frac{1}{2-F(Q)} \int_Q^\infty f(q)  [1-G(\bar s(\phi^Q)-q)] dq = \frac{k}{1+k},
    \end{equation}
    so that in particular the expression on the left-hand side of (\ref{eqn-normal-ban}) is constant in $Q$. 
    
    Now, fix $Q'>Q$ and denote $d=Q'-Q>0$. We show that $\bar s(\phi^{Q'})-d < \bar s(\phi^{Q})$ must hold. Indeed, if that was not the case then  
    \begin{eqnarray}\label{eqn-normal-ban-2}
       \int_{Q'}^\infty f(q)  [1-G(\bar s(\phi^{Q'})-q)] dq = \int_{Q'}^\infty f(q)  [1-G(\bar s(\phi^{Q'})-d-(q-d))] dq = \notag\\
        \int_{Q}^\infty f(q'+d)  [1-G(\bar s(\phi^{Q'})-d-q')] dq' \le \int_{Q}^\infty f(q'+d)  [1-G(\bar s(\phi^Q)-q')] dq', 
    \end{eqnarray}
    where the second equality is a change of variable $q'=q-d$, and the inequality is per our assumption that $\bar s(\phi^{Q'})-d \ge \bar s(\phi^{Q})$. Combining (\ref{eqn-normal-ban}) and (\ref{eqn-normal-ban-2}) then gives
        $$\frac{1}{2-F(Q)} \int_Q^\infty f(q)  [1-G(\bar s(\phi^Q)-q)] dq \le  \frac{1}{2-F(Q')} \int_{Q}^\infty f(q+d)  [1-G(\bar s(\phi^{Q})-q)] dq,$$
    or
    \begin{equation}\label{eqn-normal-ban-3}
    \int_Q^\infty \left[\frac{f(q+d)}{2-F(Q')} - \frac{f(q)}{2-F(Q)} \right](1-G(\bar s(\phi^Q)-q)) dq \ge 0.
     \end{equation}

    On the other hand, we also have that 
    $$\int_Q^\infty \left[\frac{f(q+d)}{2-F(Q')} - \frac{f(q)}{2-F(Q)} \right]dq = \left[ \frac{1-F(Q')}{2-F(Q')} -  \frac{1-F(Q)}{2-F(Q)} \right] <0.$$
    Combining this with the assumption that the ratio $\frac{f(q+d)}{f(q)}$ is decreasing, it implies that there exists $\bar q\in [Q,\infty)$ such that $\frac{f(q+d)}{2-F(Q')} - \frac{f(q)}{2-F(Q)}$ is positive on $[Q,\bar q)$ and negative on $(\bar q, +\infty)$. Thus, replacing the increasing function $1-G(\bar s(\phi^Q)-q)$ by the constant $1-G(\bar s(\phi^Q)-\bar q)$ in (\ref{eqn-normal-ban-3}) makes the integrand pointwise larger. Hence,
    \begin{eqnarray*}
        0 &\le& \int_Q^\infty \left[\frac{f(q+d)}{2-F(Q')} - \frac{f(q)}{2-F(Q)} \right](1-G(\bar s(\phi_Q)-q)) dq \le \\
        & & (1-G(\bar s(\phi_Q)-\bar q)) \int_Q^\infty \left[\frac{f(q+d)}{2-F(Q')} - \frac{f(q)}{2-F(Q)} \right] dq =\\ 
        & & (1-G(\bar s(\phi_Q)-\bar q)) \left[ \frac{1-F(Q')}{2-F(Q')} -  \frac{1-F(Q)}{2-F(Q)} \right] <0,
    \end{eqnarray*}
    a contradiction. 
    
    We conclude that $\bar s(\phi^{Q'})-d < \bar s(\phi^{Q})$ and therefore that
    $$W(Q',\phi^{Q'})=1-G(\bar s(\phi^{Q'})-Q') = 1-G(\bar s(\phi^{Q'})-d-(Q'-d)) > 1-G(\bar s(\phi^{Q})-Q) = W(Q,\phi^Q).$$

    Since the left-hand side of (\ref{eqn-steady-state-eq}) is increasing and the right-hand side is decreasing there can be at most one equilibrium. Since $\frac{1-k}{2k}< \frac{1}{k(1-\delta)}$, Theorem \ref{thm-steady-state} and the assumption of the proposition guarantee that an equilibrium exists.

\bigskip

\noindent \textbf{Proof of Proposition \ref{prop-noisy-review-large-V}}

    Fix $\bar s>-\infty$. The proof follows the footsteps of the proof of Corollary \ref{coro-comparison-large-V}. As we have already argued, $\lim_{V\uparrow +\infty} Q_0=-\infty$, so it is sufficient to show that there is some $\underline{Q}$ such that for all $V$ large enough any steady state equilibrium cutoff $Q$ given $\bar s$ satisfies $Q>\underline{Q}$.

    The first step is to obtain a characterization of the best response to a given recurrent distribution of submissions $\phi$. Just as in Lemma \ref{lemma-best-response}, if $v(\phi)\le k$ then the unique best response is to always submit ($Q=-\infty$), and if $v(\phi)> k$ then the unique best response threshold $Q\in \RR$ is characterized by 
    $$x(Q,\phi,\bar s) > 0 ~~~~\textit{and}~~~~ W(Q,\phi)V - (1-W(Q,\phi))C = \delta(1-\delta)G(\bar s-Q)x(Q,\phi,\bar s).$$
    The latter equation can be rewritten as 
    \begin{equation}\label{eqn-best-response-noisy-review}
    W(Q,\phi) = \frac{C+\delta(1-\delta)G(\bar s-Q)x(Q,\phi,\bar s)}{C+V}.
    \end{equation}

    Next, the volume of submissions in the steady state associated with cutoff $Q$ is given by $\frac{1-F(Q)}{1+Ban(Q,\bar s)}$. Note that for $Q=-\infty$ this volume is larger than $k$ by the assumption in the proposition. Moreover, if $Q>-\infty$ is a steady state equilibrium then this volume must also be larger than $k$, since otherwise the probability of winning for any quality is 1, contradicting the fact that qualities below $Q$ do not apply. Therefore, the volume of submissions in any steady state equilibrium is greater than $k$. Just as in the proof of Theorem \ref{thm-steady-state}, it then follows from \textbf{A2} that the steady state winning probability of a researcher under equilibrium $Q$ is given by
    $$Win \left(Q,\frac{1}{1+Ban(Q,\bar s)}f^Q \right) = k(1+Ban(Q,\bar s)).$$
    Therefore, if $Q$ is an equilibrium then the payoff of an eligible researcher is given 
    $$x\left(Q,\frac{1}{1+Ban(Q,\bar s)}f^Q ,\bar s\right) = \frac{k(1+Ban(Q,\bar s))V-(1-F(Q)- k(1+Ban(Q,\bar s)))C}{(1-\delta)(1+\delta Ban(Q,\bar s))}.$$
    Plugging this into the right-hand side of (\ref{eqn-best-response-noisy-review}), we get that the following must hold in equilibrium:
    \begin{eqnarray*}
    & &   W\left(Q,\frac{1}{1+Ban(Q,\bar s)}f^Q\right) = \\
    & &   \frac{C(1+\delta Ban(Q,\bar s))+\delta G(\bar s-Q)\left[k(1+Ban(Q,\bar s))V-(1-F(Q)- k(1+Ban(Q,\bar s)))C\right]}{(C+V)(1+\delta Ban(Q,\bar s))}.
    \end{eqnarray*}

    Now, when $V$ grows to infinity, the limit of the right-hand side of the above equation is 
    $$\frac{\delta k G(\bar s-Q)(1+Ban(Q,\bar s))}{1+\delta Ban(Q,\bar s)}\ge \delta k G(\bar s-Q).$$
    Also, since by assumption $\frac{1}{1+Ban(-\infty,\bar s)}>k$, it follows as in the proof of Proposition \ref{prop-benchmark} that $\lim_{Q\downarrow -\infty} W\left(Q,\frac{1}{1+Ban(Q,\bar s)}f^Q\right)=0$. Let $\underline{Q}$ be small enough so that 
    $$ W\left(Q,\frac{1}{1+Ban(Q,\bar s)}f^Q\right) < \delta k G(\bar s-Q)$$ 
    for all $Q<\underline{Q}$. Then, for all sufficiently large $V$, any equilibrium cutoff $Q$ must be greater than $\underline{Q}$.
    

\bigskip

\noindent \textbf{Proof of Proposition \ref{prop-noisy-review-comparison}}

    Fix $\bar s$ and let $Q$ be an equilibrium cutoff. To simplify the notation let $\phi(q) := \frac{1}{1+Ban(Q,\bar s)}f^{Q}(q)$ be the quality distribution of submissions in the steady states equilibrium given $\bar s$. Recall that the quality distributions of winners are then given by
    $$h(q)=\phi(q)(1-G(\bar s(\phi)-q)) = \frac{1}{1+Ban(Q,\bar s)}f^{Q}(q)(1-G(\bar s(\phi)-q))$$
    and 
    $$h_0(q)=f^{Q_0}(q)\left(1-G\left(\bar s\left(f^{Q_0}\right)-q\right)\right).$$

    We start by proving part (i) of the proposition. Suppose then that $Q>Q_0$. If $q<Q$ then $h(q)=0$ so clearly $h(q)\le h_0(q)$. For $q\ge Q$ we have that $h(q)\le h_0(q)$ if and only if
    \begin{equation}\label{eqn-condtional-unconditional-normal}
        \frac{1}{1+Ban(Q, \bar s)} \le \frac{1-G\left(\bar s\left(f^{Q_0}\right)-q\right)}{1-G(\bar s (\phi)-q)}.
    \end{equation}
    Notice that $f^{Q_0}(q)\ge \phi(q)$ (with strict inequality for $q>Q_0$) and therefore that $\bar s\left(f^{Q_0}\right) > \bar s (\phi)$. Since it is assumed that $G$ has increasing hazard rate, it follows that the right-hand side of (\ref{eqn-condtional-unconditional-normal}) is monotonically increasing in $q$. Furthermore, the left-hand side of (\ref{eqn-condtional-unconditional-normal}) is strictly smaller than 1, while the right-hand side converges to 1 as $q\to +\infty$. Thus, there exists some $\bar Q>Q$ such that $h(q)\le h_0(q)$ for $q>\bar Q$ and the reverse inequality holds whenever $q\in [Q,\bar Q]$.  

    Consider now case (ii) where $Q\le Q_0$. Here, too, (\ref{eqn-condtional-unconditional-normal}) is necessary and sufficient for $h(q)\le h_0(q)$ whenever $q\ge Q_0$. We distinguish between two cases: First, suppose that $\bar s\left(f^{Q_0}\right)\le \bar s (\phi)$. In this case the right-hand side of (\ref{eqn-condtional-unconditional-normal}) is at least 1 for all $q$, implying that $h(q)\le h_0(q)$ on the entire range $q\in (Q_0,\infty)$, and hence that $h_0$ first-order dominates $h$. Second, suppose that $\bar s\left(f^{Q_0}\right)> \bar s (\phi)$. Then as in case (i) above the right-hand side of (\ref{eqn-condtional-unconditional-normal}) increases in $q$ and converges to 1 as $q\to \infty$. Therefore, $h(q)< h_0(q)$ for all sufficiently large $q$ and these two functions can cross at most once in the range $q\in (Q_0,\infty)$. This again implies that $h_0$ first-order dominates $h$ and the proof is complete.

\bigskip

\noindent \textbf{Proof of Lemma \ref{lemma-increasing-T}}

    Fix $\phi$ as in the lemma. First, note that the set of cutoffs $Q$ at which $x_t(Q,\phi)$ is positive is independent of $t$. Second, it is immediate to verify that for any $Q$ in this set the product $(1-\delta^t)x^t(Q,\phi)$ is increasing in $t$ (recall that this is exactly the loss in payoff due to exclusion). 

    Fix $t>t'$, let $Q$ be the best response to $\phi$ with $t$-period exclusion, and let $Q'$ be the best response with $t'$-period exclusion. Then 
    \begin{eqnarray*}
    W(Q,\phi) = \frac{C+\delta(1-\delta^t)x_t(Q,\phi)}{C+V+\delta(1-\delta^t)x_t(Q,\phi)} & \ge & \frac{C+\delta(1-\delta^t)x_t(Q',\phi)}{C+V+\delta(1-\delta^t)x_t(Q',\phi)} > \\
    & & \frac{C+\delta(1-\delta^{t'})x_{t'}(Q',\phi)}{C+V+\delta(1-\delta^{t'})x_{t'}(Q',\phi)} = W(Q',\phi),
    \end{eqnarray*}
    where the first equality is by (\ref{eqn-best-response-T}), the first inequality is by the optimality of $Q$ given $t$, the next inequality follows from $t>t'$ and $x_t(Q',\phi)>0$, and the last equality is again from (\ref{eqn-best-response-T}). Since $W$ is strictly increasing in $Q$ it follows that $Q>Q'$.

\bigskip

\noindent \textbf{Proof of Proposition \ref{prop-T-large}}

    Recall first that from Theorem \ref{thm-steady-state-T} we must have $Q_t\le Q^*$ for all $t$. Now, if $\{Q_t\}$ does not diverge to $-\infty$ then for some $M$ small enough there are infinitely many elements of the sequence in the interval $[M,Q^*]$, and hence there is a convergent subsequence in this interval. With abuse of notation let $\{Q_t\}$ be said subsequence and let $\bar Q$ be its limit point.  
    
    The share of eligible researchers $\frac{1+tk}{1+t(1-F(Q_t))}$ converges to $\frac{k}{1-F(\bar Q)}$ and therefore the sequence of submission quality distributions $\{\phi^{Q_t}_t\}$ converges to $\bar \phi:= \frac{k}{1-F(\bar Q)}f^{\bar Q}$. Note that $v(\bar \phi)=k$. Thus, the left-hand side of (\ref{eqn-steady-state-eq-T}) satisfies
    $$\lim_t W(Q_t, \phi^{Q_t}_t) = W(\bar Q, \bar \phi)=1,$$
    where the first equality is by continuity \textbf{A1} and the second is by \textbf{A2}. On the other hand, at the right-hand side of (\ref{eqn-steady-state-eq-T}) we have
\begin{eqnarray*}
     \lim_t \frac{(1+tk)C+k\delta(1+\delta+\ldots+\delta^{t-1}) (1+t(1-F(Q_t)))V}{(1+tk)C+(1+tk)[1+\delta(1+\delta+\ldots+\delta^{t-1})(1-F(Q_t))]V} =\\
    \frac{(1-\delta)kC+k(\delta -\delta F(\bar Q))V}{(1-\delta)kC+k(1-\delta F(\bar Q))V},
\end{eqnarray*}
    which is bounded away from one for $\bar Q\in (-\infty, Q^*]$. Therefore, for large enough $t$ equality (\ref{eqn-steady-state-eq-T}) cannot hold, which is a contradiction to the assumption that $Q_t$ is an equilibrium cutoff for all $t$.

    The rest of the claims in the proposition immediately follow. Indeed, for (i), $\lim_t \frac{1+tk}{1+t(1-F(Q_t))} = k$. For (ii), note that the sequence $\phi^{Q_t}_t$ pointwise converges to $kf$ and therefore (by Lebesgue's dominated convergence theorem) $d(\phi^{Q_t}_t,kf)\to 0$. Similarly, the sequence $W(\cdot,\phi^{Q_t}_t)\phi^{Q_t}_t$ also converges pointwise to $kf$ (since $W(q,\phi^{Q_t}_t)\to 1$ for each $q$ by \textbf{A1} and \textbf{A2}), proving (iii). 

\bigskip

\noindent \textbf{Proof of Proposition \ref{prop-compare-types}}

    Suppose that the distribution of qualities submitted every period is $\phi$ and that $v(\phi)>k$. Then, just as in Lemma \ref{lemma-best-response}, the unique best response cutoff $Q$ for a type $i$ researcher is characterized by $x_i(Q,\phi)>0$ and
    \begin{equation}\label{eqn-types-best-response}
    W(Q,\phi)= \frac{C+\delta(1-\delta)x_i(Q,\phi)}{C+\delta(1-\delta)x_i(Q,\phi)+V}.
    \end{equation}
    If $v(\phi)\le k$ then $Q=-\infty$ is the unique best response for both types. 

    Now, suppose that $(Q_1^H,Q_1^L)$ is a steady state equilibrium, and let $(\alpha_H,\alpha_L)$ be the corresponding steady state eligibility shares. Denote $\phi=\alpha_H f_H^{Q_1^H} + \alpha_L f_L^{Q_1^L}$ the resulting equilibrium quality distribution of submissions. Note that we can't have $v(\phi)\le k$, since then the best responses are $Q_1^H=Q_1^L=-\infty$, which by Definition \ref{def-fixed-point-types} implies that (\ref{eqn-fixed-point-types}) must hold for both $i=H,L$. Summing up these two equations gives $\alpha_H+\alpha_L=\frac{1+k}{2}>k$, which contradicts the assumption that $v(\phi)\le k$. 

    Therefore, using (\ref{eqn-types-best-response}), to prove the proposition it is sufficient to show that $x_H(Q_1^H,\phi)\ge x_L(Q_1^L,\phi)$, with strict inequality when $F_H(q)<F_L(q)$ everywhere. Denote $p=1-F_L(Q_1^L)$ and let $\bar Q$ be such that $1-F_H(\bar Q)=p$. Consider the cdf's $\tilde F_L,\tilde F_H$ defined by 
    $$\tilde F_L(q) = \left\{ \begin{array}{ll}
                        \frac{F_L(q)-F_L(Q_1^L)}{p} & \textit{ if } ~~ q \ge Q_1^L,\\
                        0 & \textit{ if } ~~q < Q_1^L,\\
    \end{array} \right.$$
    and 
    $$\tilde F_H(q) = \left\{ \begin{array}{ll}
                        \frac{F_H(q)-F_H(\bar Q)}{p} & \textit{ if } ~~ q \ge \bar Q,\\
                        0 & \textit{ if } ~~q < \bar Q.\\
    \end{array} \right.$$
    In words, $\tilde F_L$ ($\tilde F_H$) is the quality distribution of an $L$ ($H$) type conditional on the quality being at least $Q_1^L$ ($\bar Q$). It is immediate that if $F_H$ first-order dominates $F_L$, then $\tilde F_H$ first-order dominates $\tilde F_L$. Further, if $F_H(q)<F_L(q)$ for all $q$, then $\tilde F_H(q)<\tilde F_L(q)$ for all $q>Q_1^L$. Since $W(q,\phi)$ is strictly increasing in $q$ by \textbf{A3}, we have that   
    $$Win_H(\bar Q,\phi) = p\EE_{q\sim \tilde F_H}[W(q,\phi)]\ge p\EE_{q\sim \tilde F_L}[W(q,\phi)] = Win_L(Q_1^L,\phi),$$
    and the inequality is strict when $F_H<F_L$. Therefore, from (\ref{eqn-steady-state-payoff-types}) it follows that 
    \begin{eqnarray*}
        x_H(\bar Q,\phi) &=& \frac{Win_H(\bar Q,\phi)V-(p-Win_H(\bar Q,\phi))C}{(1-\delta)[1 + \delta(p-Win_H(\bar Q,\phi))]} \ge \\
                          & &  \frac{Win_L(Q_1^L,\phi)V-(p-Win_L(Q_1^L,\phi))C}{(1-\delta)[1 + \delta(p-Win_L(Q_1^L,\phi))]} = x_L(Q_1^L,\phi),
    \end{eqnarray*}
    with strict inequality when $F_H<F_L$. 

    Finally, since $Q_1^H$ is type's $H$ best response to $\phi$, we also have that $x_H(Q_1^H,\phi) \ge x_H(\bar Q,\phi)$. Combining this with the above inequality we obtain $x_H(Q_1^H,\phi)\ge x_L(Q_1^L,\phi)$, and with strict inequality when $F_H<F_L$. This proves the proposition

\end{document}